\documentclass[final,twoside,11pt]{entics} 
\usepackage{enticsmacro}
\usepackage{graphicx}
\usepackage[all]{xy}
\usepackage{amsmath,amssymb,float,tcolorbox,stmaryrd}
\usepackage{dashrule}
\usepackage{mathpartir,bussproofs}
\usepackage{cleveref}
\usepackage{mathtools}
\usepackage{mathrsfs}

\newcommand\new[0]{\reflectbox{\ensuremath{\mathsf{N}}}}

\newcommand{\tf}[1]{\ensuremath \mathtt{#1}}

\newcommand{\nom}[1]{\mathscr{#1}}

\newcommand{\nalg}[1]{\mathfrak{#1}}

  {\gdef\scalefactor{#1}\begin{center}\proofSkipAmount \leavevmode}%
  {\scalebox{\scalefactor}{\DisplayProof}\proofSkipAmount \end{center} }

\newcommand{\Int}[3]{\llbracket #1 \rrbracket^{#2}_{#3}} 

\newcommand{\COREf}{{\sf CORE}_{\fix{}}}

\newcommand{\A}{\mathbb{A}}

\newcommand{\C}{\ensuremath\mathtt{C}}

\newcommand{\D}{\mathsf{D}}

\newcommand{\F}{\mathbb{F}}

\newcommand{\V}{\mathbb{V}}

\newcommand{\T}{{\tt T}}

\newcommand{\Fix}[1]{{\tt pfix}(#1)}

\newcommand{\fix}[1]{\curlywedge_{#1}} 

\newcommand{\atom}{\mathtt{atom}}

\newcommand{\abs}{\mathtt{abs}}

\newcommand{\nfix}[1]{\not\mathrel{\curlywedge}_{#1}} 

\newcommand{\falphaeq}[1]{\overset{\fix{}}{\approx}_{\tt  #1}} 

\newcommand{\dom}[1]{\mathtt{dom}(#1)} 

\newcommand{\perm}[1]{\mathtt{perm}(#1)} 

\newcommand{\lin}[1]{\overline{#1}} 

\newcommand{\var}[1]{\mathtt{Var}(#1)} 

\newcommand{\act}{\cdot} 

\newcommand{\supp}[2]{\mathtt{supp}_{#1}(#2)} 

\newcommand{\id}{\mathtt{id}} 

\newcommand{\atm}[1]{\mathtt{Atm}(#1)} 

\newcommand{\Perm}[1]{{\mathtt{Perm}}(#1)} 

\newcommand{\fresh}{\#} 

\newcommand{\pow}[2]{\mathcal{P}_{#1}(#2)} 

\newcommand{\pos}[1]{\mathtt{Pos}(#1)}

\newcommand{\stupsilon}[2]{\Upsilon_{#1,#2}}

\newcommand{\uprob}{{\sf Pr}}
\newcommand{\eprob}[1]{{\sf Pr}_{\tt #1}}

\newcommand{\daniele}[1]{\textcolor{cyan}{#1}}

\newcommand{\ali}[1]{\textcolor{red!50!blue!50}{#1}}

\newtheorem{atheorem}{Theorem}[section]
\newtheorem{alemma}{Lemma}[section]
\def\conf{MFPS 2024}
\def\myconf{mfps40} 	
\volume{4}			
\def\volu{4}			
\def\papno{6}			
\def\mydoi{14777}%

\def\lastname{Caires-Santos, Fern\'andez, Nantes-Sobrinho} 
\def\runtitle{Strong Nominal Semantics for Fixed-Point Constraints} 
\def\copynames{A. Caires-Santos, M. Fern\'andez, D. }
\def\copyname{\ \ Nantes-Sobrinho }    
\begin{document}
\begin{frontmatter}
  \title{A Strong Nominal Semantics for Fixed-point Constraints} 						
\thanks[ALL]{Partially supported by Brazilian FAP-DF Project DE 00193.00001175/2021-11, Brazilian CNPq Project Universal 409003/2021-2}   
  \author{Ali K. Caires-Santos \thanksref{a}\thanksref{myemail}}
   \author{Maribel Fern\'andez\thanksref{b}\thanksref{coemail}}
   \author{Daniele  Nantes-Sobrinho\thanksref{a}\thanksref{c}\thanksref{ccoemail}}	

   \address[a]{Department of Mathematics, University of Bras\'ilia, Bras\'ilia, Brazil}  							
   \thanks[myemail]{Email: \href{mailto:A.K.C.R.Santos@mat.unb.br} {\texttt{\normalshape
    A.K.C.R.Santos@mat.unb.br}}} 
  \address[b]{Department of Informatics, King's College London, London, UK} 
  \thanks[coemail]{Email: \href{mailto:maribel.fernandez@kcl.ac.uk} {\texttt{\normalshape
        maribel.fernandez@kcl.ac.uk}}}
\address[c]{Department of Computing, Imperial College London, London, UK} 
  \thanks[ccoemail]{Email:  \href{mailto:d.nantes-sobrinho@imperial.ac.uk} {\texttt{\normalshape
        d.nantes-sobrinho@imperial.ac.uk}}}
\begin{abstract} 
Nominal algebra
includes $\alpha$-equality and freshness constraints on nominal terms endowed with a nominal set semantics that facilitates reasoning about languages with binders. Nominal unification is decidable and unitary, however, its extension with equational axioms such as Commutativity (which is finitary) is no longer finitary unless permutation fixed-point constraints are used.  
In this paper, we extend the notion of nominal algebra by introducing fixed-point constraints and provide a sound 
semantics using strong nominal sets.  We show, by providing a counter-example, that the class of nominal sets is not a sound denotation for this extended nominal algebra. 
To recover soundness we propose two different formulations of nominal algebra, one obtained by restricting to a class of fixed-point contexts that are in direct correspondence with freshness contexts and another obtained by using a different set of derivation rules.
\end{abstract}
\end{frontmatter}

\section{Introduction}
\label{sec:introduction}

Nominal algebra~\cite{DBLP:journals/logcom/GabbayM09,DBLP:journals/logcom/Gabbay09} provides a sophisticated framework to interpret and reason algebraically about systems with binding. The term language of nominal algebra is an extension of first-order terms, called nominal terms~\cite{DBLP:journals/tcs/UrbanPG04}. Nominal syntax includes \emph{names} (representing object-level variables, or \emph{atoms}), along with meta-level variables (called just \emph{variables}) and  constructs to specify binding (nominal \emph{abstraction}). The $\alpha$-equivalence relation between nominal terms is defined using name \emph{permutations}  and an auxiliary \emph{freshness} relation. 

Gabbay and Mathijssen~\cite{DBLP:journals/logcom/GabbayM09} define a sound and complete proof system for nominal algebra, where derivation rules are subject to freshness conditions $a\fresh t$ (read “$a$ is fresh for $t$”~\cite{DBLP:journals/tcs/UrbanPG04}).
Nominal algebra with freshness constraints has a semantics in the class of nominal sets (sets with a finitely supported permutation action) which have a rich structure to reason about names, the permutation of names, and name binding.  This approach to the specification of syntax with binding has many advantages: Matching and unification modulo $\alpha$-equivalence are decidable~\cite{DBLP:journals/tcs/UrbanPG04} and moreover there are efficient algorithms to compute unique most general unifiers (pairs of a freshness context and a substitution) for solvable unification problems~\cite{DBLP:journals/jcss/CalvesF10,DBLP:conf/rta/LevyV10}. These algorithms  have been used to implement rewriting,  functional languages and logic programming languages (see, e.g., \cite{DBLP:journals/iandc/FernandezG07,DBLP:conf/icfp/ShinwellPG03,DBLP:conf/iclp/CheneyU04}). Disunification problems are also decidable~\cite{DBLP:conf/lsfa/Ayala-RinconFNV20}, as well as the more general class of equational problems~\cite{DBLP:conf/fossacs/Ayala-RinconFNV21,DBLP:journals/mscs/Ayala-RinconSFS21}.

Since applications in theorem proving and rewrite-based programming languages often require the use of associative-commutative operators, nominal unification algorithms have been extended to deal with equational axioms such as Associativity (A) and Commutativity (C)~\cite{AYALARINCON20193}. C-unification and AC-unification are decidable and finitary, however, although nominal C-unification and nominal AC-unification are decidable, they are not finitary when solutions are expressed using substitutions and freshness constraints.
To overcome this limitation, fixed-point constraints $\pi\fix{} t$ (read “the name permutation $\pi$ fixes $t$”) should be used: indeed,  nominal unification modulo C and AC theories is finitary if solutions are expressed using substitutions, freshness and  fixed-point constraints; moreover, freshness constraints can in turn be expressed using fixed-point constraints~\cite{DBLP:journals/lmcs/Ayala-RinconFN19}. 


The relation between freshness and fixed-point equations was initially proposed by Pitts~\cite{book/Pitts}, by the equivalence $a\# x \iff \new c. (a \ c)\cdot x=x$, where $a$ and $c$ are atoms and $x$ is an element of a nominal set. The equivalence uses the $\new$ (`new') quantifier that quantifies over new/fresh names. Thus, the equivalence means that $a$ is fresh for $x$ if and only if swapping $a$  for any new name $c$ in $x$ does not alter $x$, in other words, the action of the permutation $(a \ c)$ fixes $x$. The fixed-point constraints $\pi \fix{} t$ are more general, since $\pi$ is not restricted to be of the form $(a\ c)$ for a new name $c$. For example, depending on the chosen model, a fixed-point constraint $(a \ b)\fix{} X$ does not necessarily imply  $a\#X$. Thus, the semantics of general fixed-point constraints differs from the semantics of freshness constraints and it needs to be determined so that the meaning of the developments in nominal techniques using fixed-point constraints can be established.


In this paper we address this gap. We consider nominal algebras that use fixed-point constraints instead of freshness constraints.  Initially, the semantics  for nominal algebras with fixed-point constraints was expected to use nominal sets as carriers. However, through the analysis presented in this paper, we show that this assumption is not valid by providing counterexamples.
We also show that soundness can be recovered by using a  subclass of the class of nominal sets, called strong nominal sets,  or by strengthening the derivation rules using a restricted class of fixed-point contexts and the $\new$ quantifier from nominal logic. 

Summarising, our main contributions are:
\begin{enumerate}
    \item A novel semantic interpretation of general fixed-point constraints.     
    \item The definition of a notion of \emph{strong nominal algebra} as well as notions of strong axioms and theories.    
    \item A counter-example showing that the class of nominal sets is not a sound denotation for nominal theories using general fixed-point constraints.
   \item A proof that strong nominal sets are a sound denotation for nominal theories with fixed-point constraints. 
    \item An alternative formulation of a nominal theory with fixed-point constraints, whose denotational semantics is the whole class of nominal sets.
    \item An analysis of the correspondences between approaches with freshness and fixed-point constraints. 
\end{enumerate}

The above results open the way for a generalisation of the algorithms used to check nominal equational problems  to the case where the signature includes associative and commutative operators. Recall that equational problems are first-order formulas where the only predicate is equality. They have numerous applications in programming and theorem proving (see, e.g., \cite{DBLP:journals/jsc/ComonL89,DBLP:conf/birthday/Comon91,DBLP:journals/aaecc/Fernandez92,FernandezM:ACcom,KirchnerC:antpmm}). 
Extensions to nominal syntax with $\alpha$-equality and freshness constraints have already been studied~\cite{DBLP:conf/fossacs/Ayala-RinconFNV21}, but in order to consider syntax modulo equational theories and preserve the properties of the algorithms (in particular to preserve finitary unification) fixed-point constraints are needed, together with a corresponding notion of nominal algebra that can provide sound denotations for the problems. In future work, we will develop algorithms to solve nominal equational problems modulo associative and commutative axioms. 

\paragraph*{Organisation.} Section~\ref{sec:prelim} contains the basic definitions and notations used throughout the paper. Section~\ref{sec:semantics} establishes the semantics of fixed-point constraints, as well as 
the failure of soundness for standard derivation rules. Section~\ref{sec:recovering1} characterises a class of models for which soundness holds. Section~\ref{sec:recovering} discusses alternative approaches to recover soundness. Applications of the results are discussed in Section~\ref{sec:applications}. Section~\ref{sec:conclusion} concludes the paper. 

\section{Preliminaries}\label{sec:prelim}

\paragraph{Syntax.}

Fix disjoint countably infinite collections $\mathbb{A} = \{a,b,c,\ldots\}$ of {\em atoms} and  $\V = \{X,Y,Z,\ldots\}$ of {\em unknowns} or {\em variables}. We consider a {\em signature} $\Sigma$, finite set of {\em term-formers} ($\tf{f},\tf{g},\ldots$) - disjoint from  $\A$ and $\V$ - such that  each $\tf{f}\in \Sigma$ has a fixed arity, say $n$, which is a non-negative integer; we write $\tf{f}:n$ to denote that $\tf{f}$ has arity $n$.  Unless stated otherwise, we follow Gabbay’s \emph{permutative convention} \cite{DBLP:journals/fac/GabbayP02}: atoms $a, b$ range permutatively over $\A$ therefore saying two atoms $a$ and $b$ are distinct is unnecessary. We use {\em permutations}  to be able to rename atoms, for the purpose of $\alpha$-conversion. A {\em finite} permutation $\pi$ of atoms is a bijection $\A \to \A$ such that the set $\dom{\pi} := \{a \in\A \mid \pi(a) \neq a\}$ is finite; we say that $\pi$ has {\em finite domain/support}. Write $\id$ for the {\em identity permutation}, $\pi\circ \pi'$ for the {\em composition} of  permutations $\pi$ and $\pi'$, and $\pi^{-1}$ for the {\em inverse} of $\pi$. This makes permutations into a group; we will write $\Perm{\A}$ to denote the group of finite permutations of atoms.

The syntax of \emph{nominal terms} is defined inductively by the following grammar:
\[ s,t::= a \mid \pi\cdot X \mid [a]t \mid \tf{f}(t_1,\ldots, t_n) \]
where $a$ is an {\em atom}, $[a]t$ denotes the \emph{abstraction} of the atom $a$ over the term $t$, ${\tt f} (t_1,\ldots,t_n)$ is the \emph{application} of the term-former ${\tt f}:n$ to a tuple $(t_1,\ldots,t_n)$ and $\pi\act X$ is a {\em suspension}, where $\pi$ is an atom permutation. Intuitively, $\pi\cdot X$ denotes that $X$ will get substituted for a term and then $\pi$ will permute the atoms in that term accordingly (i.e., the permutation $\pi$ is suspended). Suspensions of the form $\id\act X$ will be represented simply by $X$. The set of all nominal terms formed from a signature $\Sigma$ will be denoted by $\F(\Sigma,\A,\V)$ or simply $\F(\Sigma,\V)$ since $\A$ will be fixed. 

Abstractions will be used to represent binding operators. A fundamental example is the lambda operator from the $\lambda$-calculus. Terms in the $\lambda$-calculus will be represented using the signature $\Sigma = \{\tf{lam}: 1, \tf{app} : 2\}$. If one consider $\lambda$-variables as atoms, $\lambda$-terms can be inductively generated by the grammar:
\[
    e::= a\mid \tf{app}(e,e)\mid \tf{lam}([a]e)
\]
To sugar the notation, we write $\tf{app}(s, t)$ as $(s\ t)$ and $\tf{lam}([a] s)$ as $\lambda a.s$.

The lambda operator is just one example amongst many others. We can also use abstractions to represent various other binding operators, such as the existential ($\exists$) and universal ($\forall$) quantifiers of first-order logic, or  the ($\nu$) operator in the $\pi$-calculus.

We say that a nominal term is {\em ground} (in a fixed signature $\Sigma$) when it does not mention unknowns (and mentions only term-formers in $\Sigma$). These are inductively defined by the grammar:
\[
    g,h ::= a \mid [a]g \mid \tf{f}(g_1,\ldots,g_n)
\]
where here $\tf{f}$ ranges over elements of $\Sigma$. The set of all ground terms will be denoted by $\F(\Sigma)$. We define the set of \emph{free names} of a ground term $g$ inductively by $\tf{fn}(a):= \{a\}, \tf{fn}(\tf{f}(g_1,\ldots,g_n)) = \bigcup_{i=1}^n\tf{fn}(g_i),$ and $\tf{fn}([a]g) := \tf{fn}(g)\setminus\{a\}$.

Write $t \equiv u$ for {\em syntactic identity} of terms. We define $a\in t$ (reads: ``$a$ occurs in $t$'') inductively by:
    \begin{center}
        \AxiomC{\phantom{$\pi(a)\neq a$}}\UnaryInfC{$a\in a$}\DisplayProof \quad   \AxiomC{\phantom{$a\in t$}}\UnaryInfC{$a\in [a]t$}\DisplayProof \quad \AxiomC{$a\in t$}\UnaryInfC{$a\in [b]t$}\DisplayProof \quad  \AxiomC{$a\in \dom{\pi}$}\UnaryInfC{$a\in \pi\act X$}\DisplayProof \quad \AxiomC{$a\in t_i \quad (1\leq i\leq n)$}\UnaryInfC{$a\in {\tt f}(t_1,\ldots,t_n)$}\DisplayProof.
    \end{center}
When $a\in t$ does not hold, we write $a \notin t$ and say that ``$a$ does not occur in $t$''. Furthermore, we define $\atm{t} := \{a \mid a \in t\}$, the set of atoms that occur in $t$.

We define $X\in t$ (reads: ``$X$ occurs in $t$'') inductively by:
    \begin{center}
        \AxiomC{\phantom{$\pi(a)\neq a$}}\UnaryInfC{$X\in \pi\act X$}\DisplayProof \quad  \AxiomC{$X\in t$}\UnaryInfC{$X\in [b]t$}\DisplayProof \quad  \AxiomC{$X\in t_i \quad (1\leq i\leq n)$}\UnaryInfC{$X\in {\tt f}(t_1,\ldots,t_n)$}\DisplayProof.
    \end{center}
When $X \in t$ does not hold, we write $X \notin t$ and say that ``$X$ does not occur in $t$''. Furthermore, we define $\var{t} := \{X \mid X \in t\}$, the set of variables that occur in $t$. Terms whose set of variables is empty, that is, there is no occurrence of variables in their structures, are {\em ground terms}.

Write $(a\ b)$ for the permutation that {\em swaps} $a$ for $b$, i.e., it maps $a$ to $b$, $b$ to $a$ and all other $c$ to themselves, take $(a \ a) = \id$. To swap atoms $a$ and $b$ in a term $t$, we use the swapping $(a \ b)$ and write $(a \ b)\act t$ to denote the action of replacing $a$ for $b$ and vice-versa in $t$. An atom permutation is built as a list of swappings.  Given two permutations $\pi$ and $\rho$, the permutation $\pi^\rho = \rho\circ\pi\circ\rho^{-1}$ denotes the {\em conjugate} of $\pi$ with respect to $\rho$. \emph{Permutation action} on terms is given by: $\pi \act a = \pi(a), \pi \act (\pi' \act X) =
(\pi \circ \pi') \act X, \pi \act ([a]t) = [\pi(a)](\pi \act t)$, and $\pi \act {\tt f}(t_1, \ldots , t_n) = {\tt f}(\pi \act t_1, \ldots \pi \act t_n)$. 

\emph{Substitutions}, denoted by $\sigma$, are finite mappings from variables to terms, which extend homomorphically over terms.
Atoms can be abstracted by an abstraction  operator as in $[a]t$ but cannot be instantiated by a substitution, while unknowns cannot be abstracted but can be instantiated by a substitution.

\paragraph*{Constraints, Judgements and Theories.}

Differently from the standard nominal approach~\cite{book/Pitts,DBLP:journals/fac/GabbayP02,DBLP:journals/logcom/GabbayM09}, we do not use freshness constraints\footnote{Pairs of the form $a\#t$ meaning that ``$a$ does not occur free in $t$''.}, instead we consider the following two kinds of constraints:
\begin{itemize}
\item     A {\em fixed-point constraint} is a pair $\pi\fix{} t$ of a permutation $\pi$ and a nominal term $t$. A fixed-point constraint $\pi \fix{} X$ is called {\em primitive}. Write $\Upsilon$ and $\Psi$ for finite sets of primitive fixed-point constraints and call them {\em fixed-point contexts}.  $\Upsilon|_X$ denotes the restriction of $\Upsilon$ to $X$ and $\perm{\Upsilon|_X}$ is the set of permutations such that $\pi \fix{} X$ is in $\Upsilon$.

Intuitively, $\pi\fix{}t$ means that $t$ is a fixed-point of $\pi$ (i.e.,  the permutation  $\pi$ fixes $t$):  $\pi\cdot t=t$.

\item An $\alpha$-{\em equality constraint} (or just equality) is a pair $t=u$ where $t$ and $u$ are nominal terms.

Equality constraints will be used to state that two terms are \emph{provably equal}.
\end{itemize}

Accordingly, we consider two judgement forms:
    \begin{itemize}
        \item A {\em fixed-point judgement form} $\Upsilon \vdash \pi\fix{} t$ is a pair of a fixed-point context $\Upsilon$ and a fixed-point constraint $\pi\fix{} t$. 
        \item An {\em equality judgement form} $\Upsilon \vdash t=u$ is a pair of a fixed-point context $\Upsilon$ and an equality constraint $t=u$.
    \end{itemize}
    We may write $\emptyset \vdash \pi\fix{} t$ as$~\vdash \pi\fix{} t$. Similarly for $\emptyset \vdash t=u$.

\begin{definition}[Theory]
    A {\em theory} $\T = (\Sigma,Ax)$ consists of a signature $\Sigma$ and a (possibly) infinite set of equality judgement forms $Ax$, known as the {\em axioms}.
\end{definition}

\begin{example}
    $\COREf$ represents a theory with no axioms. More specifically, for any signature $\Sigma$, there exists one such theory. Hence, $\COREf^{\Sigma} = (\Sigma,\emptyset)$.
\end{example}

Derivation rules for fixed-point and $\alpha$-equality judgements are given in \Cref{fig:equality-rules1}. 
 Rules ($\fix{} {\bf a}$) and ($\fix{} {\tt f}$) are straightforward. In the rule ($\fix{} {\bf var}$), the condition $\dom{\pi^{\pi'^{-1}}} \subseteq \dom{\perm{\Upsilon|_X}}$ means that the atoms affected by the permutation $\pi^{\pi'^{-1}}$ are in the domain of the permutations that fix $X$. 
The rule ($\fix{} {\bf abs}$) states that $[a]t$ is a fixed-point of $\pi$ if $\pi$ fixes $(a\ c_1)\act t$, where $c_1$ is a fresh atom for $t$ (the latter is indicated by the fact that $(c_1 \ c_2)$ fixes $t$).
Rules $({\bf refl}), ({\bf symm})$ and $({\bf tran})$ ensure that equality is an equivalence relation,
and the rules $({\bf cong}[])$ and $({\bf cong}{\tt f})$ ensure that it is a congruence. The $({\bf ax_{\Upsilon'\vdash t=u}})$ rule instantiates axioms in derivations. In rule $({\bf ax})$ we use the notation $(\pi\act \Upsilon')\sigma = \{\pi'^\pi\fix{} \pi\act X\sigma \mid \pi'\fix{} X\in\Upsilon'\}$. The $({\bf perm})$ rule states that if swapping $a$ with a new name $c_1$ fixes $t$ and similarly swapping $b$ with a new name $d_1$ fixes $t$, then swapping $a$ and $b$ does not affect $t$.  
Finally, $({\bf fr})$ (read top-down) is a strengthening rule. Note that the $({\bf perm})$ rule together with the $(\fix{})$ rules implements $\alpha$-equivalence.
For example, we can prove $\vdash [a]a = [b]b$ as follows:
{\small 
\begin{mathpar} 
\inferrule*[RIGHT=({\bf perm})]{\inferrule*[RIGHT=(\mbox{$\fix{}$}{\bf abs})]{\inferrule*[RIGHT=(\mbox{$\fix{}$}{\bf a})]{(b \ c_4)\act b = c_4 \quad \text{ and } \quad (a \ c_1)(c_4) = c_4}{\vdash (a\ c_1)\fix{} (b \ c_4) \cdot  b}}{\vdash (a\ c_1)\fix{} [b]b} \\ \inferrule*[RIGHT=(\mbox{$\fix{}$}{\bf abs})]{\inferrule*[RIGHT=(\mbox{$\fix{}$}{\bf a})]{(b \ c_3)\act b = c_3 \quad \text{ and } \quad (b \ c_2)(c_3) = c_3}{\vdash (b\ c_2)\fix{} (b \ c_3)\cdot  b} }{\vdash (b\ c_2)\fix{} [b]b }}{\vdash (a\ b) \cdot [b]b = [b]b}
\end{mathpar}
}

In the following, $\Upsilon\vdash \pi\fix{} t$ denotes that $\pi\fix{} t$ is derivable using  at most the elements of $\Upsilon$ as assumptions.  We write $\Upsilon\not\vdash \pi\fix{}t$ when $\Upsilon \vdash \pi\fix{}t$ is not derivable. Similarly, for equality judgments we write $\Upsilon \vdash_{\T} t=u$ if $t = u$ can be derived  i) using assumptions from $\Upsilon$; ii)  for each occurrence of $({\bf ax}_{\Upsilon'\vdash t=u})$ used in the derivation, $(\Upsilon' \vdash t=u)\in Ax$; iii) only terms in the signature $\Sigma$ are mentioned in the derivation.

      \begin{figure}[!t]
     \centering

         \hrule
{\small 
      \begin{tabular}{lr}
         & \\
         \AxiomC{$\pi(a) = a$}
         \RightLabel{$(\fix{} \textbf{a})$}
         \UnaryInfC{$\Upsilon \vdash \pi\fix{} a$}
         \DisplayProof & \AxiomC{$\Upsilon \vdash \pi\fix{} t_1 \quad \ldots \quad \Upsilon \vdash \pi\fix{} t_n$}
         \RightLabel{$(\fix{} {\tt f})$}
         \UnaryInfC{$\Upsilon \vdash \pi\fix{} {\tt f}(t_1,\ldots,t_n)$}
         \DisplayProof\\[0.5cm]
         \AxiomC{$\dom{\pi^{\pi'^{-1}}}\subseteq \dom{\perm{\Upsilon|_X}}$}
         \RightLabel{$(\fix{} \textbf{var})$}
         \UnaryInfC{$\Upsilon \vdash \pi\fix{} \pi'\act X$}
         \DisplayProof & \AxiomC{$\Upsilon,\lin{(c_1 \ c_2)\fix{} \var{t}} \vdash \pi\fix{} (a \ c_1)\act t$}
         \RightLabel{$(\fix{} \textbf{abs})$}
         \UnaryInfC{$\Upsilon \vdash \pi\fix{} [a]t$}
         \DisplayProof \\
         & \\
     \end{tabular}\\

     \hdashrule{\textwidth}{.3pt}{3pt} \\

\bigskip

     \begin{tabular}{c}

         \AxiomC{}\RightLabel{$({\bf refl})$}\UnaryInfC{$\Upsilon \vdash t=t$}\DisplayProof \quad \AxiomC{$\Upsilon \vdash t=u$ }\RightLabel{$({\bf symm})$}\UnaryInfC{$\Upsilon \vdash u=t$}\DisplayProof \quad \AxiomC{$\Upsilon \vdash t=u$ }\AxiomC{$\Upsilon \vdash u=v$}\RightLabel{$({\bf tran})$}\BinaryInfC{$\Upsilon \vdash t=v$}\DisplayProof
         \\[0.5cm]
         \AxiomC{$\Upsilon \vdash (\pi\act\Upsilon')\sigma$}\RightLabel{$({\bf ax_{\Upsilon'\vdash t=u}})$}
         \UnaryInfC{$\Upsilon \vdash \pi\act t\sigma = \pi\act u\sigma$}
         \DisplayProof \qquad 
         \AxiomC{$\Upsilon \vdash t=u$}\RightLabel{$({\bf cong}[])$}\UnaryInfC{$\Upsilon \vdash [a]t = [a]u$}\DisplayProof \quad   \AxiomC{$\Upsilon \vdash t=u$}\RightLabel{$({\bf cong}{\tt f})$}\UnaryInfC{$\Upsilon \vdash {\tt f}(\ldots,t,\ldots) = {\tt f}(\ldots,u,\ldots)$}\DisplayProof\\[0.5cm]
         \AxiomC{$\Upsilon,\pi\fix{} X \vdash t=u \quad (\dom{\pi}\subseteq\dom{\perm{\Upsilon|_X}})$}\RightLabel{$({\bf fr})$}\UnaryInfC{$\Upsilon \vdash t=u$}\DisplayProof\\[0.5cm]
         \AxiomC{$\Upsilon,\lin{(c_1 \ c_2)\fix{}\var{t}} \vdash (a \ c_1)\fix{} t$}
         \AxiomC{$\Upsilon,\lin{(d_1 \ d_2)\fix{}\var{t}} \vdash (b \ d_1)\fix{} t$}\RightLabel{$({\bf perm})$}
         \BinaryInfC{$\Upsilon \vdash (a \ b)\act t = t$}\DisplayProof\\
         \\
     \end{tabular}
     }
\hrule
       \caption{Derivation rules; $c_1,c_2,d_1,d_2$ are fresh names. Also, $\lin{\pi \fix{}\var{t}} = \{\pi\fix{} X\mid X\in\var{t}\}$.}
       \label{fig:equality-rules1}
 \end{figure}


    


\subsection{Nominal Sets}

We recall some basic definitions on nominal sets~\cite{book/Pitts}.

 A {\em $\Perm{\A}$-set}, denoted by $\nom{X}$, is a pair $(|\nom{X}|,\act)$ consisting of an {\em underlying set} $|\nom{X}|$ and a {\em permutation action} $\cdot $, which is a group action on $|\nom{X}|$, i.e., an operation $\act:\Perm{\A}\times |\nom{X}|\to |\nom{X}|$ such that $\id\act x = x$ and $\pi\act (\pi'\act x) = (\pi\circ\pi')\act x$, for every $x\in |\nom{X}|$ and $\pi,\pi'\in \Perm{\A}$.  Sometimes, we will write $\pi\act_{\nom{X}} x$, when we want to make $\nom{X}$ clear. 

For $B \subseteq \A$ write $\Fix{B} = \{\pi\in\Perm{\A}\mid \forall a\in B.~\pi(a) = a\}$, that is, $\Fix{B}$ is the set of permutations that fix pointwise the elements of $B$. A set of atomic names $B\subseteq \A$ {\em supports} an element $x \in |\nom{X}|$ when for all permutations $\pi\in\Perm{\A}$, $\pi\in\Fix{B} \Longrightarrow \pi\act x = x$. Additionally, we say that $B$ {\em strongly supports} $x\in|\nom{X}|$ if for all permutations $\pi\in\Perm{\A}$, $\pi\in\Fix{B} \Longleftrightarrow \pi\act x = x$. 

A {\em nominal set} is a $\Perm{\A}$-set $\nom{X}$ all of whose elements are finitely supported. $\nom{X,Y,Z}$ will range over nominal sets. A nominal set is {\em strong} if every element is strongly supported by a finite set\footnote{The class of strong nominal sets is a subclass of the class of nominal sets}. 
The {\em  support} of an element $x\in|\nom{X}|$ of a nominal set $\nom{X}$ is defined as $\supp{}{x} = \bigcap\{B\mid \text{$B$ is finite and supports $x$}\}$. This implies that $\supp{}{x}$ is the {\em least} finite support of $x$. In the case where $\nom{X}$ is strong, if $x$ is strongly supported by a finite set $B\subseteq \A$, then $B = \supp{}{x}$.

For any nominal sets $\nom{X,Y}$, call a map $f: |\nom{X}|\to |\nom{Y}|$ {\em equivariant} when $\pi\act f(x) = f(\pi\act x)$ for all $\pi\in\Perm{\A}$ and $x\in |\nom{X}|$. In this case we write $f:\nom{X}\to\nom{Y}$. For instance, any constant map is easily an equivariant map. 


\begin{example}[Some simple nominal sets]\label{ex:nominal-and-strong-sets}
    \begin{enumerate}
        \item \label{ex:atom-set} The $\Perm{\A}$-set $(\A,\act)$ with the action $\pi\act_\A a = \pi(a)$ is a nominal set and $\supp{}{a} = \{a\}$.

        \item \label{ex:pow-fin-set} Consider the set $\pow{\tt fin}{\A} = \{B\subset\A \mid B \text{ is finite}\}$. Then the $\Perm{\A}$-set $(\pow{\tt fin}{\A},\act)$ with the action $\pi\act_{\pow{\tt fin}{\A}} B = \{\pi\act_\A a\mid a \in B\}$ is a nominal set and $\supp{}{B} = B$. Observe that $\pow{\tt fin}{\A}$ is not strong because if we take $B = \{a,b\}$ and $\pi = (a \ b)$, then $\pi\act B = B$ but $\pi\notin \Fix{B}$.

        \item \label{ex:singleton-set} The singleton set $\{\star\}$ equipped with the action $\pi\act \star = \star$ is a strong nominal set and $\supp{}{\star} = \emptyset$. 

        \item The set $\A^* = \bigcup \{a_1\cdots a_n\mid \forall i,j\in\{1,\ldots,n\}. a_i\in\A\wedge(j\neq i\Rightarrow a_j\neq a_i)\}$, that is, the set of finite words over distinct atoms, is a strong nominal set when equipped with the permutation action given by $\pi\act (a_1\cdots a_n) = \pi(a_1)\cdots\pi(a_n)$, and $\supp{}{a_1\cdots a_n} = \{a_1,\ldots,a_n\}$.

        \item \label{ex:ground-algebra} $\F(\Sigma)$ with the action $\pi\act a \equiv \pi(a), \pi \act [a]t \equiv [\pi(a)]\pi\act t$ and $\pi \act \tf{f}(t_1,\ldots,t_n)\equiv \tf{}(\pi\act t_1,\ldots,\pi \act t_n)$ forms a nominal set and $\supp{}{g} = \atm{g}$ for all $g\in \F(\Sigma)$. The relation $\sim$ defined by $g \sim g'$ iff $~\vdash_\T g = g'$ is an equivariant equivalence relation. The set $\F(\Sigma)/_{\sim}$ is a nominal set and $\supp{}{\lin{g}} = \bigcap\{\supp{}{g'} \mid g'\in\lin{g}\}$ for any $\lin{g} \in \F(\Sigma)/_{\sim}$ (see Proposition 2.30~\cite{book/Pitts}). We denote $\F(\Sigma)/_{\sim}$ by $\lin{\F}(\T,\Sigma)$. In the case where $\T = \emptyset$, the relation $\sim$ is the $\alpha$-equivalence relation between ground terms and $\supp{}{\lin{g}} = \tf{fn}(g)$ for all $\lin{g}\in \lin{\F}(\emptyset,\Sigma)$.
    \end{enumerate}
\end{example}

\section{Semantics}\label{sec:semantics}
In this section we provide a semantics for fixed-point constraints using nominal sets, and denote its semantic interpretation as $\fix{\tt sem}$. We then build upon the semantic definition of $\fix{}$ and define concepts such as (strong) $\Sigma$-algebras, valuations, interpretations, models, etc. These are necessary ingredients to  determine validity of judgements in terms of $\fix{\tt sem}$, and consequently, establish soundness of nominal algebra with fixed-point constraints.

Below we overload the symbol $=$ to denote equality between elements of a nominal set.


\subsection{Semantics for \texorpdfstring{$\fix{}$}{fix}}

 The semantics of fixed-point constraints $\fix{\tt sem}$ differs from $\#_{\tt sem}$ and it is defined in terms of fixed-point equations, and not in terms of the support. For an element of a nominal set  $x\in \nom{X}$, if  $ a\notin \supp{}{x}$,  we write 
$a\#_{\tt sem}\ x$. We refer to \cite{DBLP:journals/logcom/GabbayM09} for more details about the semantics of nominal algebra defined by freshness constraints. Below we define $\fix{\tt sem}$.

\begin{definition}\label{def:fix_sem}
    Let $\nom{X}$ be a nominal set, $x\in |\nom{X}|, a\in\A,$ and $\pi\in \Perm{\A}$. We write $\pi\fix{\tt sem} x$ to denote that $\pi\act x = x$.
\end{definition}

\begin{lemma}\label{lem:dom-supp-empty}
    Let $\nom{X}$ be a nominal set and $x\in|\nom{X}|$. If $\dom{\pi}\cap\supp{}{x} = \emptyset$ then $\pi\fix{\tt sem} x$.
\end{lemma}


\begin{proof}
    Direct from the definition of support.\qed
\end{proof}

However, the other direction does not hold in general: if we take  $x = \{a,b\}$ and $\pi= (a \ b)$ as in Example \ref{ex:nominal-and-strong-sets}(\ref{ex:pow-fin-set}), $\pi\fix{\tt sem} x$ but $\dom{\pi}\cap\supp{}{x} = \{a,b\}$. The converse does hold for strong nominal sets.

\begin{lemma}\label{thm:support-char}
    Let $\nom{X}$ be a nominal set and $x\in|\nom{X}|$. Then
     \[
    a\in \supp{}{x} \iff \{c\mid (a\ c) \fix{\tt sem}x\}~\text{ is  finite}
    \]
\end{lemma}

    

\begin{proof}
    The left-to-right implication follows by showing that $\{c\mid (a\ c) \fix{\tt sem}x\}\subseteq\supp{}{x}$. The right-to-left implication follows by contradiction and Lemma~\ref{lem:dom-supp-empty}.\qed
\end{proof}

\begin{theorem}\label{thm:strong-theorem}
     Let $\nom{X}$ be a strong nominal set, $x\in |\nom{X}|$, and $\pi,\gamma_1,\ldots,\gamma_n\in\Perm{\A}$. If $\gamma_i\fix{\tt sem} x$ for every $i=1,\ldots, n$, and $\dom{\pi}\subseteq\bigcup_{1\leq i \leq n} \dom{\gamma_i},$ then $\pi\fix{\tt sem} x$.
\end{theorem}

\begin{proof}
    For each $i=1,\ldots,n$ the following hold:
    $$\gamma_i\fix{\tt sem} x \iff \gamma_i\cdot x=x \iff \gamma_i\in\Fix{\supp{}{x}}. $$
    The first equivalence comes from Definition~\ref{def:fix_sem} and the second from the fact that $\nom{X}$ is a strong nominal set. Then by the inclusion $\dom{\pi}\subseteq\bigcup_{1\leq i \leq n} \dom{\gamma_i}$, we deduce that $\pi\in\Fix{\supp{}{x}}$ which implies, by definition,  that $\pi\fix{\tt sem} x$.\qed
\end{proof}

Note that Theorem~\ref{thm:strong-theorem} is not true in the class of nominal sets. For instance, take $\nom{X}$ as $\pow{\tt fin}{\A}$, $x = \{a,b\}$, $\pi = (a \ c)$, 
$\gamma_1 = (a\ b)$, $\gamma_2 = (c \ d)$.
Then, $\gamma_1\fix{\tt sem} x,\gamma_2\fix{\tt sem} x$, and $\dom{\pi} \subseteq\dom{\gamma_1}\cup\dom{\gamma_2}$,  but $\pi\nfix{\tt sem} x$.

\subsection{Strong Nominal \texorpdfstring{$\Sigma$}{Sigma}-algebras}

We can now define a semantics for nominal algebra with fixed-point constraints in (strong) nominal sets.

\begin{definition}[(Strong) $\Sigma$-algebra]\label{def:sigma-algebra-with-nominal-set}
    Given a signature $\Sigma$, a {\em (strong) nominal $\Sigma$-algebra} $\nalg{A}$ consists of:
    \begin{enumerate}
        \item A (strong) nominal set $\nom{A} = (|\nom{A}|,\act)$ - the domain.

        \item An equivariant map $\atom^{\nalg{A}}:\A\to |\nom{A}|$ to interpret atoms; we write the interpretation $\atom(a)$ as $a^{\nalg{A}}\in |\nom{A}|$.
        
        \item An equivariant map $\abs^{\nalg{A}}:\A\times |\nom{A}|\to |\nom{A}|$ such that $\{c\in \A\mid (a \ c)\nfix{\tt sem} \abs^{\nalg{A}}(a,x)\}$ is finite, for all $a\in \A$ and $x\in|\nom{A}|$ (by Lemma~\ref{thm:support-char} this is equivalent to saying that $a\notin \supp{}{\abs^{\nalg{A}}(a,x)}$); we use this to interpret abstraction.

        \item An equivariant map $f^{\nalg{A}}:|\nom{A}|^n \to |\nom{A}|$, for each $\tf{f}:n$ in $\Sigma$ to interpret term-formers.
    \end{enumerate}
    We use $\nalg{A,B,C}$ to denote (strong) $\Sigma$-algebras.
\end{definition}

\begin{remark}
    The standard definition of the condition in item 3 via freshness constraints is $a\fresh_{\tt sem} \abs^{\nalg{A}}(a,x)$ (see, for instance, \cite{DBLP:journals/logcom/GabbayM09}). Here, we use the following equivalence (valid in nominal sets):
    \[\begin{aligned}
        a\fresh_{\tt sem} \abs^{\nalg{A}}(a,x) &\Leftrightarrow a\notin \supp{}{\abs^{\nalg{A}}(a,x)} & (\text{Definition 4.2,~\cite{DBLP:journals/logcom/GabbayM09}})\\
        &\Leftrightarrow\{c\in \A\mid (a \ c)\not\mathrel{\curlywedge_{\tt sem}} \abs^{\nalg{A}}(a,x)\} \text{ is finite} &(\text{Consequence of Lemma~\ref{thm:support-char}})
        \end{aligned}
    \]
\end{remark}

A {\em valuation} $\varsigma$ in a (strong) $\Sigma$-algebra $\nalg{A}$ maps unknowns $X\in\V$ to elements
$\varsigma(X) \in |\nom{A}|$. Below we define an equivariant function $\Int{\act}{\nalg{A}}{\varsigma}: T(\Sigma,\A,\V)\to |\nom{A}|$ to interpret nominal terms w.r.t. a valuation $\varsigma$.

\begin{definition}[Interpretation]\label{def:interpretation-of-terms}
    Let $\nalg{A}$ be a (strong) $\Sigma$-algebra. Suppose that $t \in T(\Sigma, \A, \V)$ and consider a valuation $\varsigma$ in $\nalg{A}$. The {\em interpretation} $\Int{t}{\nalg{A}}{\varsigma}$, or just $\Int{t}{}{\varsigma}$, if $\nalg{A}$ is understood, is defined inductively by:

        \begin{tabular}{l@{ \hspace{1cm} }l}
            $\Int{a}{\nalg{A}}{\varsigma} = a^{\nalg{A}}$ & $\Int{ \pi\act X}{\nalg{A}}{\varsigma} = \pi\act \varsigma(X)$ \\
            $\Int{{\tt f}(t_1,\ldots,t_n)}{\nalg{A}}{\varsigma} = f^{\nalg{A}}(\Int{t_1}{\nalg{A}}{\varsigma},\ldots, \Int{t_n}{\nalg{A}}{\varsigma})$ &  $\Int{[a]t}{\nalg{A}}{\varsigma} = \abs(a^{\nalg{A}},\Int{t}{\nalg{A}}{\varsigma})$.
        \end{tabular}
\end{definition}

\begin{lemma}\label{lem:equivariant-interpretation}
    The term-interpretation map $\Int{\act}{\nalg{A}}{\varsigma}$ is equivariant, that is, $\pi\act \Int{t}{\nalg{A}}{\varsigma} = \Int{\pi\act t}{\nalg{A}}{\varsigma}$ for any permutation $\pi$ and term $t$.
\end{lemma}


\begin{proof}
    The proof is by induction on the structure of the term $t$.\qed
\end{proof}

        
        



\begin{definition}[Context and Judgement Validity]
    For any (strong) $\Sigma$-algebra $\nalg{A}$:
        \begin{itemize}
           \item  $\Int{\Upsilon}{\nalg{A}}{\varsigma}$ is {\it valid} iff   $\pi\fix{\tt sem} \Int{X}{\nalg{A}}{\varsigma}$, for each $\pi\fix{} X\in\Upsilon$;
           \item  $\Int{ \Upsilon \vdash \pi\fix{} t }{\nalg{A}}{\varsigma}$ is {\it valid}  iff   $\Int{ \Upsilon}{\nalg{A}}{\varsigma}$ (valid) implies $\pi\fix{\tt sem} \Int{t}{\nalg{A}}{\varsigma}$;
          \item   $\Int{ \Upsilon \vdash t = u}{\nalg{A}}{\varsigma}$  is {\it valid}  iff $\Int{ \Upsilon}{\nalg{A}}{\varsigma}$ (valid) implies $\Int{t}{\nalg{A}}{\varsigma} = \Int{u}{\nalg{A}}{\varsigma}$.
        \end{itemize}
More generally, $\Int{ \Upsilon \vdash \pi\fix{} t }{\nalg{A}}{}$  is {\it valid} iff $\Int{ \Upsilon \vdash \pi\fix{} t }{\nalg{A}}{\varsigma}$ is valid for all valuations $\varsigma$. Similarly for  $\Int{ \Upsilon \vdash t = u}{\nalg{A}}{}$.
\end{definition}


Then a (strong) model of a theory is an interpretation that validates its axioms as follows:

\begin{definition}[(Strong) Model]
    Let  $\T = (\Sigma,Ax)$ be a theory. A {\em (strong) model of} $\T$ is a (strong) $\Sigma$-algebra ${\cal A}$ such that
    \begin{align*}
            \Int{\Upsilon \vdash t = u}{\nalg{A}}{} \text{ is valid for every axiom } \Upsilon \vdash t = u \text{ in $Ax$ and every valuation $\varsigma$.}
    \end{align*}
\end{definition}



We are now ready to define validity with respect to a theory.

\begin{definition}
    For any theory $\T$, define {\em (strong) validity with respect to} $\T$ for judgement forms as follows:
    \begin{itemize}
        \item Write $\Upsilon \vDash_\T \pi\fix{} t$ (resp. $\Upsilon \vDash_\T^s \pi\fix{} t$) iff $\Int{ \Upsilon \vdash \pi\fix{} t}{\nalg{A}}{}$ is valid for all models $\nalg{A}$ of $\T$ (resp.\ all strong models of $\T$).
        \item Write $\Upsilon \vDash_\T t=u$ (resp. $\Upsilon \vDash_\T^s t=u$) iff $\Int{ \Upsilon \vdash t = u}{\nalg{A}}{}$  is valid for all models $\nalg{A}$ of $\T$ (resp.\ all strong models of $\T$).
    \end{itemize}
\end{definition}

\begin{lemma}\label{lem:singleton-interpretation}
    Define the {\em singleton interpretation} $\nalg{S}$ to have as domain the strong nominal set $\{\star\}$, with $\atom^{\nalg{S}},\abs^{\nalg{S}},f^{\nalg{S}}$ constant functions with value $\star$. Then $\nalg{S}$ is a strong $\Sigma$-algebra and a model for every theory.
\end{lemma}


\begin{proof}
    To verify that $\nalg{S}$ is a model of $\T$, check that $\Int{\Upsilon\vdash t=u}{\nalg{S}}{\varsigma}$ is valid for each axiom $\Upsilon\vdash t=u\in Ax$ and valuation $\varsigma$ (which maps all unknowns to $\star$). If $\Int{\Upsilon}{\nalg{S}}{\varsigma}$ is valid, then $\Int{t}{\nalg{S}}{\varsigma} = \Int{u}{\nalg{S}}{\varsigma} = \star$, confirming $\nalg{S}$ as a model.\qed
\end{proof}

\begin{remark}
    Lemma \ref{lem:singleton-interpretation}  shows that we do not need to require the map $\atom$ to be injective.
\end{remark}

\subsection{Counter-Example: Soundness Failure}\label{sec:counter_exam}

Soundness of a derivation system means that derivability implies validity. More precisely, in a sound system  the following should hold:

\noindent{\bf (Soundness)}   Suppose $\T=(\Sigma,Ax)$ is a theory. 
    \begin{enumerate}
        \item If $\Upsilon \vdash \pi\fix{} t$ then $\Upsilon \vDash_\T \pi\fix{} t$.

        \item  If $\Upsilon \vdash_\T t=u$ then $\Upsilon \vDash_\T t = u$.
\end{enumerate}

Note that by definition, for a judgement to be valid it has to hold in all models $\nalg{A}$ of $\T$ for every valuation $\varsigma$. We present two counter-examples that illustrate two different points where soundness fails.

\subsubsection{Rule $\fix{}{\bf var}$.} 
The problem arises from the rule $(\fix{} {\bf var})$:
         \begin{prooftree}
            \AxiomC{$\dom{\pi^{\pi'^{-1}}}\subseteq \dom{\perm{\Upsilon|_X}}$}
            \RightLabel{$(\fix{} \textbf{var})$}
            \UnaryInfC{$\Upsilon \vdash \pi\fix{} \pi'\act X$}
        \end{prooftree}

\begin{claim} There exist a $\Sigma$-algebra $\nalg{A}$, a valuation $\varsigma$ and a derivation using rule $(\fix{} {\bf var})$ such that 
$\Upsilon \vdash \pi\fix{}\pi'\act X$ but  $\Int{ \Upsilon \vdash \pi\fix{}\pi'\act X}{\nalg{A}}{\varsigma}$ is not valid.
\end{claim}
First of all, fix an enumeration of $\V = \{X_1,X_2,\ldots\}$ and $\A = \{a_1,a_2,\ldots\}$. Now, let's give the nominal set $\pow{\tt fin}{\A}$ a $\Sigma$-algebra structure. 

\begin{enumerate}
            \item Consider the domain of $\nalg{A}$ as the nominal set $\nom{A} = (\pow{\tt fin}{\A},\act) 
    $ presented in Example \ref{ex:nominal-and-strong-sets}(\ref{ex:pow-fin-set}).
            
            \item Define $\atom^\nalg{A}\colon \A\to \pow{\tt fin}{\A}$ by $\atom^\nalg{A}(a) = \{a\}$.

            \item Define $\abs^\nalg{A}:\A\times \pow{\tt fin}{\A}\to \pow{\tt fin}{\A}$ by $\abs^\nalg{A}(a,B) = B\setminus\{a\}$.

            \item For each term-former ${\tt f}\colon n\in\Sigma$, we associate a map $f^{\nalg{A}}\colon \pow{\tt fin}{\A}^n \to \pow{\tt fin}{\A}$ defined by $f^{\nalg{A}}(B_1,\ldots,B_n) = \bigcap_{i=1}^n B_i.$
        \end{enumerate}
It is easy to see that the functions $\atom^\nalg{A},\abs^\nalg{A}, f^{\nalg{A}}$ are equivariant.

Define the valuation $\varsigma: \V \to \pow{\tt fin}{\A}$ by $\varsigma(X_i) = \{a_{i},a_{i+1}\}$. Now, let's build a derivation using the rule $(\fix{} {\bf var})$ and interpret it in $\nalg{A}$. Consider $\Upsilon = \{(a_1 \ a_2)\fix{} X_1,(a_3 \ a_4)\fix{} X_1\}$ and write $\gamma_1 = (a_1 \ a_2)$ and $\gamma_2 = (a_3 \ a_4)$. Take $\pi = (a_1 \ a_3)$. Hence \[
    \dom{\pi} = \{a_1,a_3\} \subseteq \{a_1,a_2,a_3,a_4\} = \dom{\gamma_1}\cup\dom{\gamma_2}.
\]

So  $\Upsilon \vdash \pi\fix{} X_1$ is derivable. All we need to do now is to verify the validity of $\Int{\Upsilon\vdash \pi\fix{} X_1}{\nalg{A}}{\varsigma}$. In order to so, we need to prove that
        \begin{center}
            If $\Int{\Upsilon}{\nalg{A}}{\varsigma}$ is valid then $\pi\fix{\tt sem}\Int{X_1}{\nalg{A}}{\varsigma}$.
        \end{center}

        First, observe that $\Int{\Upsilon}{\nalg{A}}{\varsigma}$ is in fact valid:
        \begin{align*}
            \gamma_1\act \Int{X_1}{\nalg{A}}{\varsigma} &= \gamma_1\act \{a_1,a_2\} = \{a_1,a_2\} = \Int{X_1}{\nalg{A}}{\varsigma}\\
            \gamma_2\act \Int{X_1}{\nalg{A}}{\varsigma} &= \gamma_2\act\{a_1,a_2\} = \{a_1,a_2\} = \Int{X_1}{\nalg{A}}{\varsigma}.
        \end{align*}

        However, 
        \begin{align*}
            \pi\act\Int{X_1}{\nalg{A}}{\varsigma} &= \pi\act\varsigma(X_1) \\
            &= \pi\act\{a_1,a_2\} \\
            &= \{a_3,a_2\} \neq \{a_1,a_2\} = \varsigma(X_1) = \Int{X_1}{\nalg{A}}{\varsigma}.
        \end{align*}

\subsubsection{Non-strong Axioms.}\label{ex:commutativity-not-strong}
Let $\tf{C}$ be the commutativity axiom: $\tf{C}=\{ \vdash X+Y = Y+X\}$. Consider the $\Sigma$-algebra with domain $\lin{\F}(\C,\Sigma)$, that is the set of ground nominal terms quotiented by $\tf{C}$ and $\alpha$-equality. Let $[t]_\C$ denote the equivalence class of $t\in \F(\Sigma)$ modulo $\alpha,\C$. 

Consider the same permutations $\gamma_1,  \gamma_2, \pi$ as above, so again  we can derive $\Upsilon \vdash \pi \fix{} X_1$.

Take $\varsigma(X_1) = [a_1 + a_2]_\C$.
Then $\Int{\Upsilon}{\nalg{A}}{\varsigma}$ is  valid:
        \begin{align*}
            \gamma_1\act \Int{X_1}{\nalg{A}}{\varsigma} = [a_2+a_1]_\C =  \Int{X_1}{\nalg{A}}{\varsigma}\qquad 
            \gamma_2\act\Int{X_1}{\nalg{A}}{\varsigma} =  [a_2+a_1]_\C = \Int{X_1}{\nalg{A}}{\varsigma}.
        \end{align*}
        
However, $\pi \nfix{\tt sem} \Int{X_1}{\nalg{A}}{\varsigma}$ since $\pi\cdot \Int{X_1}{\nalg{A}}{\varsigma}=[a_3+a_2]_\C \not = [a_1+a_2]_\C =  \Int{X_1}{\nalg{A}}{\varsigma}$.

\begin{remark}
   $\lin{\F}(\C,\Sigma)$ is not a strong nominal set: note that  $(a \ b) \fix{\tt sem} [a + b]_\C$, since $(a \ b) \cdot  [a +b]_\C = [a + b]_\C$, but $(a \ b) \nfix{\tt sem} a$, i.e., $(a \ b)\notin \Fix{\supp{}{[a+b]_\C}}=\Fix{\{a,b\}}$. 
\end{remark}

\section{Recovering Soundness: Strong Models}\label{sec:recovering1}

To avoid the problems discussed in the previous section, we now focus on strong models. With these models, it is possible to ensure soundness.

\begin{theorem}[Soundness for strong $\Sigma$-algebras]\label{thm:soundness_fix}
     Suppose $\T=(\Sigma,Ax)$ is a theory. Then the following hold:
    \begin{enumerate}
        \item If $\Upsilon \vdash_\T \pi\fix{} t$ then $\Upsilon \vDash_\T^s \pi\fix{} t$.

        \item  If $\Upsilon \vdash_\T t=u$ then $\Upsilon \vDash_\T^s t = u$.
\end{enumerate}
\end{theorem}

\begin{proof}
    The proof is done by induction on the derivation of $\Upsilon \vdash \pi\fix{}t$, focusing on the last rule applied. Here, we demonstrate the most interesting case, where the last rule is ($\fix{}${\bf var}):
        \begin{prooftree}
            \AxiomC{$\dom{\pi^{\pi'^{-1}}}\subseteq \dom{\perm{\Upsilon|_X}}$}
            \RightLabel{$(\fix{} \textbf{var})$}
            \UnaryInfC{$\Upsilon \vdash \pi\fix{} \pi'\act X$}
        \end{prooftree}
        
        Suppose  that $\Int{\Upsilon}{\nalg{A}}{\varsigma}$ is valid. We need to show that $\pi\act \Int{\pi'\act X}{\nalg{A}}{\varsigma} = \Int{\pi'\act X}{\nalg{A}}{\varsigma}$.

        If $\perm{\Upsilon|_X} = \{\gamma_1,\ldots,\gamma_n\}$, then from the inclusion $\dom{\pi^{\pi'^{-1}}}\subseteq \dom{\perm{\Upsilon|_X}}$ we get
        \(
       \dom{\pi^{\pi'^{-1}}}\subseteq \bigcup_{i=1}^n \dom{\gamma_i}.
        \)

        Moreover, the validity of $\Int{\Upsilon}{\nalg{A}}{\varsigma}$ implies that $\gamma_i\fix{\tt sem} \Int{X}{\nalg{A}}{\varsigma}$  for every $\gamma_i\fix{} X\in \Upsilon$. Thus, $\pi^{\pi'^{-1}}\act \Int{X}{\nalg{A}}{\varsigma} = \Int{X}{\nalg{A}}{\varsigma}$ by Theorem \ref{thm:strong-theorem}. The result follows by noticing that $ \pi^{\pi'^{-1}}\act\Int{X}{\nalg{A}}{\varsigma} = \Int{X}{\nalg{A}}{\varsigma} \iff \pi\act \Int{\pi'\act X}{\nalg{A}}{\varsigma} = \Int{\pi'\act X}{\nalg{A}}{\varsigma}$.
\end{proof}

\subsection{Strong Axioms}

Example~\ref{ex:commutativity-not-strong} demonstrates that there exist theories $\T$ for which $\lin{\F}(\T,\Sigma)$ is not a strong nominal set. This observation rises the following question: are there theories $\T$ for which $\lin{\F}(\T,\Sigma)$ forms a strong nominal set? We propose to characterise such theories by defining a notion of {\em strong axiom}, and we will refer to them as {\em strong theories}.

In the definition below, let $<_{lex}$ be the lexicographic ordering of the positions in a term: e.g.,   $1.2 <_{lex}2.2$, $1.1<_{lex}1.2$.  

Given a term $t$, we can order the occurrences of variables in $t$ according to their position: we write $X<_t Y$ if $X$ occurs in $t$ at a position $p$ and $Y$ occurs in $t$ at a position $q$ such that $p <_{lex} q$.

We say that a term  $t$ is well-ordered if $<_t$ is a strict partial order, i.e., there is no $X,Y\in \var{t}$ such that $X<_t Y$ and $Y<_t X$. We will only consider axioms with well-ordered first-order terms.
\begin{definition}\label{def:strong_ax}
An axiom $\Upsilon\vdash t=u$ is {\em strong} if the following hold:

\begin{enumerate} 
\item \label{def:strong_ax-item1} $t$ and $u$ are first-order terms (i.e., they are built using function symbols and variables), and $\Upsilon=\emptyset$;
\item \label{def:strong_ax-item2} $u$ and $t$ are well-ordered;
\item \label{def:strong_ax-item3}  the order of the variables that occur in $t$ and in $u$ is compatible: i.e., if  $X <_t Y$ then it is not the case that $Y <_u X$. 
\end{enumerate}
Note that it could be that the variables $X$ and $Y$ occur in $t$ and not in $u$, the important fact is that there are no disagreements in the ordering in $t$ and $u$. 
\end{definition}

Condition~(\ref{def:strong_ax-item1}) excludes axioms such as $Ax=\{~\vdash a=b\}$, $Ax=\{~\vdash a+b=b+a\}$ and $Ax = \{~\vdash \tf{f}([a]X,[b]Y) = \tf{g}([b]X,[a]Y)\}$. Condition~(\ref{def:strong_ax-item2}) excludes axioms such as distributivity $\tf{D}=\{~\vdash X*(Y + Z)= X* Y + X * Z\}$.  Permutative theories, where axioms have the form $~\vdash \tf{f}(X_1,\ldots,X_n) = \tf{f}(X_{\rho(1)},\ldots, X_{\rho(n)})$, for some permutation $\rho$ of $\{1,\ldots,n\}$, do not satisfy the definition either, i.e. permutative axioms are not strong. 

\begin{example}[Strong Axioms]\label{ex:strong_axioms} The following axioms (and their combinations) are strong:
Associativity $\tf{A} = \{~\vdash \tf{f}(\tf{f}(X,Y),Z) = \tf{f}(X,\tf{f}(Y,Z))\}$. Homomorphism $\tf{Hom}=\{~\vdash \tf{h}(X+Y)=\tf{h}(X)+\tf{h}(Y)\}$. Idempotency $\tf{I}=\{~\vdash \tf{g}(X,X)=X\}$.
Neutral element $\tf{N}=\{~\vdash X * 0 = 0\}$. Left-/right-projection $\tf{Lproj}=\{~\vdash \tf{pl}(X,Y)=X\}$ and $\tf{Rproj}=\{~\vdash \tf{pr}(X,Y)=Y\}$.
\end{example}

\begin{example}[Theory $\tf{ATOM}$]
Consider the axiom $Ax=\{~\vdash a=b\}$ that generates the theory $\tf{ATOM}$. By our Definition~\ref{def:strong_ax}, this axiom is not strong.
\end{example}

\begin{theorem}
    If $\T$ is a strong theory then $\lin{\F}(\T,\Sigma)$ is a strong nominal set.
\end{theorem}


\begin{proof}
    First, as observed in Example~\ref{ex:nominal-and-strong-sets}(\ref{ex:ground-algebra}) the set $\lin{\F}(\T,\Sigma)$ is a nominal set and $\supp{}{\lin{g}} = \bigcap\{\supp{}{x}\mid x\in \lin{g}\}$ for all $\lin{g}\in \lin{\F}(\T,\Sigma)$. It remains to prove that $\supp{}{\lin{g}}$ is a strong support, that is, for all permutations $\pi$, the following holds
    \[
        \pi\fix{\tt sem} \lin{g} \Longrightarrow \pi\in\Fix{\supp{}{\lin{g}}}.
    \]
    Assume, by contradiction, that $\pi\fix{\tt sem} \lin{g}$ and $\pi\notin\Fix{\supp{}{\lin{g}}}$, that is, there is an atom $a\in\supp{}{\lin{g}}$ such that $\pi(a) \neq a$; let's assume $\pi(a) = b$. 

  From $\pi\fix{\tt sem} \lin{g}$ it follows, by definition, that  $\lin{\pi\cdot g}=\pi \cdot \lin{g}=\lin{g}$. Thus, $\vdash_{\T}\pi\cdot g=g$ (see Example~\ref{ex:nominal-and-strong-sets}(\ref{ex:ground-algebra})).  
     Now, we proceed by analysing the derivation of $\vdash_\T \pi\cdot g=g$:
     \begin{itemize}
        \item If rule {\bf (ax)} is not applied in the derivation, then $\vdash \pi\cdot g=g$. Since the application of a permutation does not change the structure of terms, it must be the case that $\pi\cdot c=c$, for every free atom $c$ of $g$. This contradicts the assumption above that $\pi(a)=b$. Therefore, $\pi\in\Fix{\supp{}{\lin{g}}}$ and the result follows.
      \item If rule {\bf (ax)} is applied in the derivation of $\vdash_\T \pi\cdot g=g$, assume it is the last step (the reasoning generalises by induction). Suppose the axiom $ \vdash t=u \in Ax$ (in $\T$) is used. Since the axiom is strong, the order of the atoms $a,b$, that are introduced by instantiation of the variables in the axiom, is preserved in  both  sides.  Thus, the order of these atoms is preserved in the derivation of $\vdash_\T \pi\cdot g=g$. This contradicts the hypothesis that $\pi(a)=b$: if $a$ occurs before $b$ in an instance of $t$ that is equal to 
      $g$ (in an application of {\bf (ax)}), then  $b$ will occur before $a$, in $\pi\cdot g$, and this contradicts the fact that $\T$ is strong.  Therefore,  $\pi\in\Fix{\supp{}{\lin{g}}}$, and the result follows.
    
     \end{itemize}
\end{proof}


\section{Recovering Soundness: Alternative Approaches}\label{sec:recovering}

In the previous section we proved soundness with respect to strong models (where carriers are strong nominal sets). Here 
our goal is to recover soundness while using a nominal set semantics. For this, we need to address the mismatch between the derivations and the notion of validity. We can proceed in different ways:
\begin{enumerate}

 \item we can change the form of the judgements and adapt the rules to ensure only the ones that are valid in all models can be derived; or 
    \item we do not change the judgements but we change the derivation rules to prevent the derivation of judgements that are not valid in all models.

\end{enumerate}

The challenge with the second approach is to ensure the system is  sufficiently powerful to derive all the valid judgements. The challenge with the first is to ensure that we do not lose expressive power. In this section we show how to address these challenges: we define two sound  systems for nominal algebra with fixed-point constraints and nominal set semantics, one  with similar rules as before but stronger judgements and  another with the same judgements as before but stronger rules.



\subsection{Using Strong Fixed-Point Contexts}\label{ssec:strong_judgements}
In this section we present an alternative proof system for  nominal algebra with fixed-point constraints. It restricts the form of fixed-point contexts in judgements, which will now be called {\it strong fixed-point contexts} and include a $\new$-quantifier~\cite{PittsA:noml} to keep track of new atoms in fixed-point constraints.


\begin{definition}[Strong Fixed-Point Context] 
A \emph{strong fixed-point context} is a finite set $\Upsilon_{A,C}$ that contains only primitive constraints of the form  $\new c.(a \ c) \fix{} X$ such that $a \in A$ and $c \in C$, for two given disjoint set of atoms $A$ and $C$. We will move the $\new$-quantifiers to the front and write $\new C. \Upsilon_{A,C}$. We will omit the indices $A$, $C$ when there is no ambiguity. As before, we will use the symbols $\Upsilon, \Psi, \ldots$ to represent strong contexts.   
 \end{definition}

Intuitively,  strong fixed-point contexts contain only fixed-point constraints that correspond directly to freshness constraints if we assume that the set $C$ contains new atoms: $\new c. (a \ c) \fix{} X$ corresponds to  $a \fresh X$.


 \begin{definition}[Strong Judgements]\label{def:strong-judgment}
  A {\emph{strong fixed-point judgement}} has the form $\new \lin{c}. (\Upsilon_{A,\lin{c_0}} \vdash \pi \fix{} t)$  where $\Upsilon_{A,\lin{c_0}}$ is a strong fixed-point context and $\lin{c_0}\subseteq \lin{c}$. Similarly, a strong  $\alpha$-equality judgement has the form  $\new \lin{c}. (\Upsilon_{A,\lin{c_0}} \vdash s \falphaeq{} t)$ where $\Upsilon_{A,\lin{c_0}}$ is a strong fixed-point context and $\lin{c_0}\subseteq \lin{c}$.
 \end{definition}
 
Note that the $\new$ quantifiers at the front of the judgement quantify the constraints in the context: for example in the judgement $\new c_1,c_2. (a\ c_1) \fix{} X, (a\ c_2) \fix{} X \vdash s \falphaeq{} t$ we have $\new c_1. (a\ c_1) \fix{} X$ in the context.


The alternative proof system is given in \Cref{fig:fixed-rules_new2}. It is a modified version of a system introduced in~\cite{DBLP:journals/lmcs/Ayala-RinconFN19}, that also used $\new$ in the derivable constraints (on the rhs of the $\vdash$), but did not use strong fixed-point contexts. In \Cref{fig:fixed-rules_new2} we adapt the notations $\dom{\perm{\Upsilon_{A,\lin{c}}|_X}} = \{a \mid (a \ c)\fix{} X\in \Upsilon_{A,\lin{c}}\}$.





\begin{figure}[!t]
     \centering

{\small 
     \begin{tabular}{lr}
         \hline
         & \\
         \AxiomC{$\pi(a) = a$}
         \RightLabel{$(\fix{} \textbf{a})$}
         \UnaryInfC{$ \new\lin{c}.\stupsilon{A}{\lin{c_0}} \vdash \pi\fix{} a$}
         \DisplayProof 
         & 
         \AxiomC{$\dom{\pi^{\rho^{-1}}}\setminus\lin{c} \subseteq \dom{\perm{\stupsilon{A}{\lin{c_0}}|_X}}$}
         \RightLabel{$(\fix{} \textbf{var})$}
         \UnaryInfC{$\new\lin{c}.\stupsilon{A}{\lin{c_0}}\vdash \pi\fix{} \rho\act X$}
         \DisplayProof \\[0.5cm]
         \AxiomC{$\new\lin{c}.\stupsilon{A}{\lin{c_0}} \vdash \pi\fix{} t_1 \quad \ldots \quad \new\lin{c}.\stupsilon{A}{\lin{c_0}} \vdash  \pi\fix{} t_n$}
         \RightLabel{$(\fix{} {\tt f})$}
         \UnaryInfC{$\new\lin{c}.\stupsilon{A}{\lin{c_0}} \vdash \pi\fix{} {\tt f}(t_1,\ldots,t_n)$}
         \DisplayProof 
         &
          \AxiomC{$\new\lin{c},c_1.\stupsilon{A}{\lin{c_0}} \vdash \pi\fix{} (a \ c_1)\act t$}
         \RightLabel{$(\fix{} \textbf{abs})$}
         \UnaryInfC{$\new\lin{c}.\stupsilon{A}{\lin{c_0}} \vdash \pi\fix{} [a]t$}
          \DisplayProof  
     \end{tabular}
         \begin{tabular}{l@{ \hspace{-3
     cm} }r}
         & \\
         \AxiomC{ }
         \RightLabel{$(\falphaeq{} \textbf{a})$}
         \UnaryInfC{$ \new\lin{c}.\stupsilon{A}{\lin{c_0}} \vdash a\falphaeq{} a$}
         \DisplayProof &            \AxiomC{$\dom{\rho^{-1}\circ\pi}\setminus\lin{c}\subseteq\dom{\perm{\stupsilon{A}{\lin{c_0}}|_X}}$}
         \RightLabel{$(\falphaeq{} \textbf{var})$}
         \UnaryInfC{$ \new\lin{c}.\stupsilon{A}{\lin{c_0}} \vdash \pi\act X \falphaeq{} \rho\act X$}
         \DisplayProof \\[0.5cm]
       \AxiomC{$ \new\lin{c}.\stupsilon{A}{\lin{c_0}} \vdash t_1 \falphaeq{} t_1' \quad \ldots \quad  \new\lin{c}.\stupsilon{A}{\lin{c_0}} \vdash t_n\falphaeq{} t_n'$}
         \RightLabel{$(\falphaeq{} {\tt f}$)}
         \UnaryInfC{$ \new\lin{c}.\stupsilon{A}{\lin{c_0}} \vdash {\tt f}(t_1,\ldots,t_n) \falphaeq{} {\tt f}(t_1',\ldots,t_n')$}
         \DisplayProof & 
         \AxiomC{$\new\lin{c}.\stupsilon{A}{\lin{c_0}} \vdash  t \falphaeq{} t'$}
         \RightLabel{$(\falphaeq{} \textbf{[a]})$}
         \UnaryInfC{$\new\lin{c}.\stupsilon{A}{\lin{c_0}} \vdash [a]t \falphaeq{} [a]t'$}
         \DisplayProof  \\[0.5cm]
         \AxiomC{$\new\lin{c}.\stupsilon{A}{\lin{c_0}} \vdash s \falphaeq{} (a \ b)\act t \qquad \new\lin{c}, c_1.\stupsilon{A}{\lin{c_0}}\vdash (a \ c_1) \fix{} t$}
         \RightLabel{$(\falphaeq{} \textbf{ab})$}
         \UnaryInfC{$\new\lin{c}.\stupsilon{A}{\lin{c_0}} \vdash [a]s \falphaeq{} [b]t$}
         \DisplayProof & \\
         &\\
         \hline
     \end{tabular}
     }
     \caption{Derivation rules for strong judgements. Here, $\lin{c}$ denotes a list of distinct atoms $c_1,\ldots,c_n$. In all the rules $\lin{c_0}\subseteq \lin{c}$.}
     \label{fig:fixed-rules_new2}
 \end{figure}

The following correctness result states that, under strong fixed-point judgements, the fixed-point constraint $\fix{}$ is still a fixed-point relation.

\begin{theorem}[Correctness]\label{thm:strong_correctness}
$\new \lin{c}.\stupsilon{A}{\lin{c_0}}\vdash \pi\fix{}t $ iff ~$\new \lin{c}.\stupsilon{A}{\lin{c_0}}\vdash \pi\cdot t\falphaeq{} t$, where $\lin{c_0}\subseteq\lin{c}$.

\end{theorem}

\begin{proof}
    The proof follows by induction on the derivations.\qed
\end{proof}

\subsubsection{From $\#$ to $\fix{}$ and back.} We recall some basic definitions of nominal algebra with freshness constraints. A freshness constraint is a pair of the form $a\#t$ where $a$ is an atom and $t$ is a nominal term, and it denotes the fact that `$a$ is fresh for $t$', that is, if $a$ occurs in $t$ it must occur abstracted. A freshness context (denoted $\Delta,\nabla$) is a finite set of primitive freshness constraints of the form $a\#X$.
An $\alpha$-equality constraint in the nominal algebra with $\#$ is denoted $s\approx t$, where $s$ and $t$ are nominal terms. A  judgement has the form $\Delta\vdash a\# t$ or $\Delta\vdash s\approx t$. The standard proof system for deriving freshness judgements is defined by the rules in~\Cref{fig:freshness_rules}.

\begin{figure}[!t]
\hrule
{\small 
\begin{mathpar}
\inferrule{\mbox{}}{\Delta\vdash a\#b}(\fresh a)
\and
\inferrule{\pi^{-1}(a)\#X\in \Delta}{\Delta \vdash a\#\pi\cdot X}(\fresh var)
\and 
\inferrule{\quad }{\Delta \vdash a\#[a]t}(\fresh [a])
\and 
\inferrule{\Delta\vdash a\#t}{\Delta \vdash a\#[b]t}(\fresh abs)
\and 
\inferrule{\Delta\vdash a\#t_1 \\ \ldots \\ \Delta\vdash a\#t_n}{\Delta \vdash a\#{\tt f}(t_1,\ldots, t_n)}(\fresh \tf{f})
\and 
\inferrule{\quad }{\Delta \vdash a\approx a}(\approx a)
\and 
\inferrule{ds(\pi,\pi')\# X\subseteq \Delta}{\Delta\vdash \pi\cdot X\approx \pi'\cdot X}(\approx var)
\and 
\inferrule{\Delta\vdash s\approx t}{\Delta \vdash [a]s\approx[a]t}(\approx [a])
\and 
\inferrule{\Delta\vdash s\approx (a\ b)\cdot  t\\ \Delta \vdash a\# t}{\Delta \vdash [a]s\approx[b]t}(\approx ab)
\and 
\inferrule{\Delta\vdash s_1\approx t_1\\ \ldots \\ \Delta\vdash s_n\approx t_n}{\Delta \vdash {\tt f}(s_1,\ldots, s_n)\approx {\tt f}(t_1,\ldots, t_n)}(\approx \tf{f})
\end{mathpar}
}
\hrule
\caption{Standard derivation rules for freshness and $\alpha$-equivalence. Here, $\tf{ds}(\pi,\pi') = \{a\in\A\mid \pi(a) \neq \pi'(a)\}$.}\label{fig:freshness_rules}
\end{figure}

Recall that the notion of an atom $a$ being fresh for an element $x$ in nominal set is defined using the \new-quantifier and a fixed-point equation  as follows~\cite{book/Pitts}:
$$ a\# x \iff \new c. (a \ c)\cdot x=x.$$

We explore this relationship and define mappings that translate strong judgements using $\fix{}$ and $\falphaeq{}$ into judgments using $\#$ and $ \approx$, and vice-versa. Below we denote by $\mathfrak{F}_\fresh$ the family of (primitive) freshness constraints, and by $\mathfrak{F}_{\fix{}}$ the family of primitive strong fixed-point constraints. The mapping $[\cdot]_{\fix{}}$ associates each primitive freshness constraint with a primitive strong fixed-point constraint; it extends to freshness contexts in a natural way
\begin{align*}
    [\cdot]_{\fix{}}\colon \mathfrak{F}_\fresh &\to \mathfrak{F}_{\fix{}}\\
    a\fresh X &\mapsto \new c_a. (a \ c_a)\fix{} X 
\end{align*}
We denote by $[\Delta]_{\fix{}}$ the image of $\Delta$ under $[\cdot]_{\fix{}}$. The mapping $[\cdot]_\fresh$ associates each primitive strong fixed-point constraint with a freshness
constraint; it extends to fixed-point contexts in a natural way.
\begin{align*}
    [\cdot]_{\fresh}\colon \mathfrak{F}_{\fix{}} &\to \mathfrak{F}_{\fresh}\\
    \new c.(a \ c)\fix{} X &\mapsto a\fresh X.
\end{align*}


Our translations map strong fixed-point judgments into freshness judgements. The next result establishes that every derivable freshness judgement can be mapped, via $[\cdot ]_{\fix{}}$ to a strong fixed-point constraint. 
\begin{theorem}[From and to \mbox{$[\cdot ]_{\fix{}}$}]\label{thm:transl1}
The following hold, for some $\lin{c}$ (possibly empty):
    \begin{enumerate}
        \item $\Delta\vdash a\# t \iff \new \lin{c},c'.[\Delta]_{\fix{}}\vdash (a \ c') \fix{} t$.
        \item $\Delta\vdash s\approx t \iff \new \lin{c}.[\Delta]_{\fix{}}\vdash s\falphaeq{} t$.
    \end{enumerate}
\end{theorem}


\begin{proof}
    The proof is by induction on the last rule applied in the derivation. We will prove only the interesting case involving the abstraction rule for item (i). The proof for item (ii) is analogous. 
    \begin{description}
         \item[$(\Rightarrow):$] 
            If the last rule applied is $(\fresh abs)$, then $t\equiv [b]t'$ and there is a derivation
            \begin{prooftree}
                \AxiomC{$\Delta \vdash a\fresh t'$}
                \RightLabel{$(\fresh abs)$}
                \UnaryInfC{$\Delta \vdash a\fresh [b]t'$}
            \end{prooftree}
            By induction, there exists a proof $\new \lin{c},c'.([\Delta]_{\fix{}})_{A,\lin{c_0}} \vdash  (a\ c') \fix{} t'$. Note now that for any fresh name $c_1$ for the derivation $\new \lin{c},c'.([\Delta]_{\fix{}})_{A,\lin{c_0}} \vdash (a\ c') \fix{} t'$, the permutation $\rho=(b \ c_1)$ is such that $\dom{\rho}\cap(\lin{c}\cup\{c'\})=\emptyset$. Using Equivariance, we have $\new \lin{c},c'.([\Delta]_{\fix{}})_{A,\lin{c_0}} \vdash (a\ c') \fix{} (b\ c_1) \cdot t'$. By the definition of $\new$, we can express this as $\new c_1.(\new \lin{c},c'.([\Delta]_{\fix{}})_{A,\lin{c_0}} \vdash (a\ c') \fix{} (b\ c_1) \cdot t')$, implying $\new \lin{c},c',c_1.([\Delta]_{\fix{}})_{A,\lin{c_0}} \vdash (a\ c') \fix{} (b\ c_1) \cdot t'$. We can put $\new c_1$ inside  because $c_1$ is a fresh atom distinct from every atom occurring in $([\Delta]_{\fix{}})_{A,\lin{c_0}}$, and thus, the result we aimed follows by rule $(\fix{} {\bf abs})$.

            \item[$(\Leftarrow):$] If the last rule applied is $(\fix{} {\bf abs})$, then we have two cases.
                
            \begin{itemize}
                    \item The first case is when $t\equiv [a]t'$. In this case, $\Delta \vdash a\fresh [a]t'$ follows trivially by rule $(\fresh [a])$.

                    \item The other case is when $t\equiv [b]t'$. In this case, we have a derivation
                    \begin{prooftree}
                        \AxiomC{$\new \lin{c},c',c_1.([\Delta]_{\fix{}})_{A,\lin{c_0}}\vdash (a \ c')\fix{} (b \ c_1)\act t'$}
                        \RightLabel{$(\fix{} {\bf abs})$}
                        \UnaryInfC{$\new \lin{c},c'.([\Delta]_{\fix{}})_{A,\lin{c_0}} \vdash (a \ c')\fix{} [b]t'$}
                    \end{prooftree}
                    By the induction hypothesis, we have $\Delta\vdash a\# (b\ c_1)\cdot t'$ and we can build the derivation
                    \begin{prooftree}
                        \AxiomC{$\Delta\vdash a\# (b\ c_1)\cdot t'$}
                        \RightLabel{$(\fresh abs)$}
                        \UnaryInfC{$\Delta\vdash a\# [c_1](b\ c_1)\cdot t'$}
                    \end{prooftree}
                    Given the equivalence $\Delta \vdash a \fresh [c_1](b\ c_1)\cdot t' \iff \Delta \vdash a \fresh (b\ c_1)\cdot [b] t'$, we can apply the Equivariance of freshness with $\rho = (b\ c_1)$. This yields: $\Delta\vdash a\#  [b] t'$.\qed
            \end{itemize}
        \end{description}
\end{proof}

The next result states that derivation of certain fixed-point constraints can be mapped, via $[\cdot ]_{\#}$, to derivations of freshness constraints, and vice-versa. Similarly for $\alpha$-equality constraints.

\begin{theorem}[From and to \mbox{$[\cdot ]_{\#}$}]\label{thm:transl2}
The following hold, for $\lin{c_0}\subseteq \lin{c}$:
    \begin{enumerate}
        \item $\new \lin{c},c_1.\stupsilon{A}{\lin{c_0}}\vdash  (a \ c_1)\fix{}t \iff [\stupsilon{A}{\lin{c_0}}]_{\#}\vdash a\# t$.
        \item $\new \lin{c}.\stupsilon{A}{\lin{c_0}}\vdash  s\falphaeq{} t \iff [\stupsilon{A}{\lin{c_0}}]_{\#}, \lin{\lin{c}\#\var{s,t}}\vdash s\approx t$.
    \end{enumerate}
\end{theorem}


\begin{proof} 
Again, in both cases the proof is by structural induction and follows by analysing the last rule applied in the derivation of the judgement. Here, we choose to prove only the interesting case involving the variable rule

                \begin{enumerate}
                    \item
                    \begin{description}
                        \item[$(\Rightarrow):$]  The last rule applied is $(\fix{} {\bf var})$: 
                    \begin{prooftree}
                        \AxiomC{$ \{\pi^{-1}(a)\} = \dom{(a \ c')^{\pi'^{-1}}}\setminus\lin{c}\cup\{c'\} \subseteq \dom{\perm{\Upsilon_{A,\lin{c_0}}|_X}}$}
                        \RightLabel{$(\fix{} {\bf var})$}
                        \UnaryInfC{$\new\lin{c},c'.\Upsilon_{A,\lin{c_0}}\vdash ( a \ c')\fix{} \pi'\act X$}
                    \end{prooftree}
                    That is, by hypothesis, $\new c_{\pi'^{-1}(a)}.(\pi'^{-1}(a) \ c_{\pi'^{-1}(a)})\fix{}\in \Upsilon_{A,\lin{c_0}}|_X$ and by the translation $[\act]_{\fresh}$ and $[\act]_{\fix{}}$ we have $\pi'^{-1}(a)\#X \in [\Upsilon_{A,\lin{c_0}}|_X]_{\fresh} \subseteq [\Upsilon_{A,\lin{c_0}}]_\fresh$. Then $[\Upsilon_{A,\lin{c_0}}]_\fresh\vdash a\fresh \pi'\act X$ by rule $(\fresh var)$. 

                    \item[$(\Leftarrow):$] The last rule applied is $(\fresh var)$: 
                    \begin{prooftree}
                        \AxiomC{$\pi^{-1}(a)\fresh X\in[\Upsilon_{A,\lin{c_0}}]_\fresh$}
                        \RightLabel{$(\fresh var)$}\UnaryInfC{$[\Upsilon_{A,\lin{c_0}}]_\fresh\vdash a\fresh \pi\act X$}
                    \end{prooftree}
                    By the translation, $\new c_{\pi'^{-1}(a)}.(\pi'^{-1}(a) \ c_{\pi'^{-1}(a)})\fix{} X\in \Upsilon_{A,\lin{c_0}}|_X$ and so $\pi^{-1}(a)$ is in $\dom{\perm{\Upsilon_{A,\lin{c_0}}|_X}}$. Thus $\{\pi^{-1}(a)\} = \dom{(a \ c')^{\pi^{-1}}}\setminus\{\lin{c},c'\}\subseteq\dom{\perm{\Upsilon_{A,\lin{c_0}}|_X}}$ and the result follows by rule $(\fix{} {\bf var})$.
                    \end{description}

                    \item 
                    \begin{description}
                        \item[$(\Rightarrow):$] The last rule applied is $(\falphaeq{} {\bf var})$:
                    
                        In this case $t \equiv \pi\act X,s\equiv \pi'\act X, \var{s,t} = \{X\}$ and there is a derivation 
                    \begin{prooftree}
                        \AxiomC{$\dom{\pi'^{-1}\circ\pi}\setminus\lin{c} \subseteq \dom{\perm{\Upsilon_{A,\lin{c_0}}|_X}}$}
                        \RightLabel{$(\falphaeq{} {\bf var})$}\UnaryInfC{$\new\lin{c}.\Upsilon_{A,\lin{c_0}}\vdash \pi\act X \falphaeq{} \pi'\act X$}
                    \end{prooftree}

                    We want to show that $[\Upsilon_{A,\lin{c_0}}]_\fresh,\lin{\lin{c}\fresh X}\vdash \pi\act X \approx \pi'\act X$. There are two cases to consider: 
                    \begin{itemize}
                    \item If $\dom{\pi'^{-1}\circ\pi}\setminus\lin{c} = \emptyset$, then $\tf{ds}(\pi,\pi') = \dom{\pi'^{-1}\circ\pi} \subseteq \lin{c}$. Consequently, $\tf{ds}(\pi,\pi')\fresh X \subseteq [\Upsilon_{A,\lin{c_0}}]_\fresh, \lin{\lin{c}\fresh X}$ holds true because $\lin{c}\fresh X\subseteq [\Upsilon_{A,\lin{c_0}}]_\fresh, \lin{\lin{c}\fresh X}$. Therefore, $[\Upsilon_{A,\lin{c_0}}]_\fresh, \lin{\lin{c}\fresh X} \vdash \pi\act X \approx \pi'\act X$ follows by $(\fresh var)$.

                    \item Now, suppose $\dom{\pi'^{-1}\circ\pi}\setminus\lin{c} \neq \emptyset$, then there are atoms in $\dom{\pi'^{-1}\circ\pi}$ that are not in $\lin{c}$. If $b\in \dom{\pi'^{-1}\circ\pi}\setminus\lin{c}$ then $b\in \dom{\perm{\Upsilon_{A,\lin{c_0}}|_X}}$, so by the definition of $[\cdot]_{\fresh}$, we have $\dom{\pi'^{-1}\circ\pi}\setminus\lin{c}\fresh X \subseteq [\Upsilon_{A,\lin{c_0}}]_\fresh$. Consequently, $\dom{\pi'^{-1}\circ\pi}\setminus\lin{c}\fresh X \subseteq [\Upsilon_{A,\lin{c_0}}]_\fresh,\lin{\lin{c}\fresh X}$. Additionally, $\lin{c}\fresh X \subseteq [\Upsilon_{A,\lin{c_0}}]_\fresh,\lin{\lin{c}\fresh X}$. Hence, $\dom{\pi'^{-1}\circ\pi} \subseteq [\Upsilon_{A,\lin{c_0}}]_\fresh,\lin{\lin{c}\fresh X}$, and the result follows by $(\fresh var)$ because $\dom{\pi'^{-1}\circ\pi} = \tf{ds}(\pi,\pi')$.
                    \end{itemize}

                    \item[$(\Leftarrow):$] The last rule applied is $(\fresh var)$:
                    \begin{prooftree}
                        \AxiomC{$\tf{ds}(\pi,\pi')\fresh X\subseteq [\Upsilon_{A,\lin{c_0}}]_\fresh,\lin{\lin{c}\fresh X}$}
                        \RightLabel{ $(\fresh var)$}\UnaryInfC{$[\Upsilon_{A,\lin{c_0}}]_\fresh,\lin{\lin{c}\fresh X}\vdash \pi\act X \approx \pi'\act X$}
                    \end{prooftree}

                    We want to show that $ \new \lin{c}.\Upsilon_{A,\lin{c_0}}\vdash\pi\act X\falphaeq{} \pi'\act X$. In order to do this, we must show that $\dom{\pi^{-1}\circ\pi}\setminus\lin{c}\subseteq \dom{\perm{\Upsilon_{A,\lin{c_0}}|_X}}$.

                    From $\tf{ds}(\pi,\pi')\fresh X\subseteq [\Upsilon_{A,\lin{c_0}}]_\fresh,\lin{\lin{c}\fresh X}$, we have $\tf{ds}(\pi,\pi')\setminus\lin{c} \fresh X \subseteq [\Upsilon_{A,\lin{c_0}}]_\fresh|_X$. By the translation $[\act]_{\fresh}$, we have that every atom in $\tf{ds}(\pi,\pi')\setminus\lin{c}$ is in $\dom{\perm{\Upsilon_{A,\lin{c_0}}|_X}}$. Then the result follows because $\tf{ds}(\pi,\pi') = \dom{\pi'^{-1}\circ\pi}$.\qed

                     \end{description}
                \end{enumerate}
\end{proof}

The following example shows that not all derivable strong fixed-point judgments can be mapped to a freshness judgement in the expected way.
\begin{example}\label{ex:fix_not_fresh}
It is easy to see that the judgement $\new c_1,c_2. (a \ c_1)\fix{} X,(b\ c_2)\fix{} X \vdash (a \ b)\fix{} X$ is derivable using the rule $(\fix{} \textbf{var})$ from Figure~\ref{fig:fixed-rules_new2} for $\dom{(a\ b)}\setminus \{c_1,c_2\}\subseteq \{a,b\}$. However, since neither $a$ nor $b$ are new names, we cannot use our translation to derive a freshness judgement as in Theorem~\ref{thm:transl2}(i). It follows from Theorem~\ref{thm:strong_correctness} (correctness) that there exists a derivation of the judgement $\new c_1,c_2. (a \ c_1)\fix{} X,(b\ c_2)\fix{} X \vdash (a \ b)\cdot X\falphaeq{}X$, and by Theorem~\ref{thm:transl2}(ii), there is a derivation of an equivalent $\alpha$-equality judgement via $\#$.
\end{example}

Since our rules in \Cref{fig:fixed-rules_new2} do not include a rule to deal with equational theories,  the following result holds for the theory $\T=\COREf$:

\begin{theorem}\label{thm:inherited soundness}
Nominal sets are sound denotational semantics for the nominal algebra with strong fixed-point constraints and judgements.
\end{theorem}
\begin{proof}
    The proofs follow from the translation $[\cdot]_{\fix{}}$ and from the fact that nominal sets are sound denotational semantics for the nominal algebra with freshness constraints~\cite{DBLP:journals/logcom/GabbayM09}.
\end{proof}



\subsection{Using More Powerful Rules}\label{ssec:group_gen}
In this section, we discuss an alternative approach to nominal algebra with fixed-point constraints that is also sound for all nominal set models. 
The changes proposed in \Cref{ssec:strong_judgements} focused on restricting the form of the fixed-point contexts, which  became equivalent to freshness contexts. Alternatively, we can keep the judgements and rules in~\Cref{fig:equality-rules1} as much as possible untouched and adapt only the problematic rules. Thus, in this subsection, we propose a change in the rule $(\fix{}\textsf{var})$ where we consider the permutation group generated by the permutations in the fixed-point context (denoted $\langle \perm{\Upsilon|_X}\rangle$) and the other rules are left as in~\Cref{fig:equality-rules1}. The new rule is called {\sc Gvar}$_{\fix{}}$:
%
%
\begin{mathpar}
    \inferrule*[right=Gvar$_{\fix{}}$]{\pi^{\rho^{-1}} \in \langle \perm{\Upsilon|_X}\rangle}{\Upsilon\vdash \pi\fix{} \rho \cdot X}
\end{mathpar}
In the proof of soundness (Theorem~\ref{thm:soundness_fix}), we used the hypothesis of $\nalg{A}$ being a {strong} nominal algebra in the proof of the variable case.  The next theorem shows that for the new set of rules, we do not need this hypothesis any longer.


\begin{theorem}[Soundness]\label{thm:soundness_fix-group-generated}
    Suppose $\T=(\Sigma,Ax)$ is a theory. Then
    \begin{enumerate}
        \item If $\Upsilon \vdash \pi\fix{} t$ then $\Upsilon \vDash_\T \pi\fix{} t$.

        \item  If $\Upsilon \vdash_\T t=u$ then $\Upsilon \vDash_\T t = u$.
\end{enumerate}
    
\end{theorem}

\begin{proof} In both cases, the proof is by induction on the rule applied. The interesting case is for the rule
{\bf {\sc Gvar}$_{\fix{}}$.}
   Then the  derivation is 
        \begin{prooftree}
            \AxiomC{$\pi^{\rho^{-1}}\in \langle\perm{\Upsilon|_X}\rangle$}
            \RightLabel{{\sc Gvar}$_{\fix{}}$}
            \UnaryInfC{$\Upsilon \vdash \pi\fix{} \rho\act X$}
        \end{prooftree}
        
        Suppose  that $\Int{\Upsilon}{\nalg{A}}{\varsigma}$ is valid. We need to show that $\pi\act \Int{\rho\act X}{\nalg{A}}{\varsigma} = \Int{\rho\act X}{\nalg{A}}{\varsigma}$.

        If $\perm{\Upsilon|_X} = \{\gamma_1,\ldots,\gamma_n\}$, then from the inclusion $\pi^{\rho^{-1}}\in \langle\perm{\Upsilon|_X}\rangle$,
         we get $\pi^{\rho^{-1}}=\gamma_1^{\alpha_1}\circ\ldots \circ \gamma_n^{\alpha_n}$, for $\alpha_1,\ldots,\alpha_n\in \mathbb{N}$.
        Moreover, the validity of $\Int{\Upsilon}{\nalg{A}}{\varsigma}$ implies that $\gamma_i\fix{\tt sem} \Int{X}{\nalg{A}}{\varsigma}$  for every $\gamma_i\fix{} X\in \Upsilon$. Thus, $\pi^{\rho^{-1}}\act\Int{X}{\nalg{A}}{\varsigma} = \Int{X}{\nalg{A}}{\varsigma}$ follows from the fact that it is in the generated group.\qed
\end{proof}

\subsubsection{Expressive Power}
The new  {\sc Gvar}$_{\fix{}}$ rule is sound, but more restrictive. The next two examples (Example~\ref{ex:newvar2} and Example~\ref{ex:newvar1}) are about judgements that are derivable using the rules in~\Cref{fig:equality-rules1}, but not with the new version using  {\sc Gvar}$_{\fix{}}$. These examples illustrate two different aspects where the group of generated permutations is more restrictive.

\begin{example} \label{ex:newvar2} It was possible to derive $(a \ b) \fix{} X, (c\ d)\fix{}X \vdash (a \ c)\fix{} X$  with the rules in~\Cref{fig:equality-rules1} including ({\bf var}$_{\fix{}}$), however, it is not possible to derive this judgement using the new  {\sc Gvar}$_{\fix{}}$ rule because $(a \ c)$ is not in the group generated by the permutations $(a \ b)$ and $(c \ d)$, i.e.,  $(a \ c)\notin \langle (a \ b),(c \ d)\rangle =\{id, (a \ b),(c \ d),(a \ b)\circ(c \ d)\}$.
\end{example}

We recall the fact that the new names are not annotated with the quantifier $\new$ in~\Cref{fig:equality-rules1}, the information about new names is relegated to the meta-language as usual when discussing languages with binders (such as the $\lambda$-calculus). As we can see below, this `sugared' notation (without \new) for new names causes some problems:

\begin{example}\label{ex:newvar1} In this example we want to mimic the derivation in~\Cref{fig:freshness_rules} of $a,b\# X \implies (a \ b) \cdot [a]X\approx [a]X$.  Suppose $\Upsilon=\{(b \ d')\fix{} X\}$. It is possible to derive $(b \ d')\fix{} X \vdash (a \ b)\act [a]X = [a]X$ using $\fix{}$ via the rules of~\Cref{fig:equality-rules1}. However, with the new rule {\sc Gvar}$_{\fix{}}$ the same judgement is not derivable: The first difficulty is that we can't inform $\Upsilon$ that $b$ is fresh, the formal way would be to have the constraint 
$\new d'. (b \ d')\fix{} X\in \Upsilon$. In the meta-language, this reduces to $(b \ d')\fix{} X\in \Upsilon$ and records the information `where $d'$ is a new name'. However, this is not satisfactory since $d'$ can be {\em any} new name, as we can see below.
{\small 
\begin{prooftree}
        \AxiomC{$(c_3\ c_1)\notin \langle (c_1 \ c_2), (c_3 \ c_4), (b\ d') \rangle$}
        \UnaryInfC{$\Upsilon,(c_1 \ c_2)\fix{}X, (c_3 \ c_4)\fix{}X \vdash (a \ c_1)\fix{} (a \ c_3)\cdot X$}
        \UnaryInfC{$\Upsilon,(c_1 \ c_2)\fix{}X \vdash (a \ c_1)\fix{} [a]X$}
         \AxiomC{$(b \ d_1)\notin \langle (d_1 \ d_2), (c_3' \ c_4'), (b\ d') \rangle$}
         \UnaryInfC{$\Upsilon,(d_1 \ d_2)\fix{}X, (c_3' \ c_4')\fix{}X \vdash (b \ d_1)\fix{} (a \ c_3')\cdot X$}
        \UnaryInfC{$\Upsilon,(d_1 \ d_2)\fix{}X \vdash (b \ d_1)\fix{} [a]X$}
         \RightLabel{$({\bf perm})$}
         \BinaryInfC{$\Upsilon \vdash (a \ b)\act [a]X = [a]X$}
\end{prooftree}
}
\noindent where $c_1,c_2,c_3,c_4,d_1,d_2,c_3',c_4'$ are new/fresh names. 

Notice that none of the sub-derivation (in the branches) are possible. However, looking closely, the conditions needed to close the above derivation, i.e., $(c_3\ c_1)\in \langle (c_1 \ c_2), (c_3 \ c_4), (b\ d') \rangle$  and $(b \ d_1)\in \langle (d_1 \ d_2), (c_3' \ c_4'), (b\ d') \rangle$, aim to check whether $b$ and $c_3$ are fresh for $X$, and this information could be derived from the context, however, using different new names ($(b \ d')$ is in $\Upsilon$ instead of $(b\ d_1)$). The problem is that the new names do not interact with the old names when generating the group of permutations.

\end{example}


We conjecture that a simple solution for both examples would consist of a combination of strong fixed-point contexts together with the generation of a larger group of permutations: including the group of permutations generated by the  new atoms of Example~\ref{ex:newvar1}, i.e., to consider the permutation group $\langle \Perm{\{c_1,c_2,c_3,c_4,c_3',c_4',d_1,d_2,d'\}},(b \ d') \rangle$. Similarly, for Example~\ref{ex:newvar2} we would need to specify the new atoms in the context and generate the permutation group of these new atoms. However, we need to check if relaxing  the group generator restriction in this way does not introduce inconsistencies.

\section{Discussion, Applications and Future Work}
\label{sec:applications}

In the standard nominal approach, using freshness constraints, nominal unification problems~\cite{DBLP:journals/tcs/UrbanPG04}, denoted $\uprob$, are finite sets of constraints of the form $s\approx^?t$ and $a\#^? t$. Solutions to (nominal) unification problems are pairs $(\Delta,\sigma)$ consisting of a freshness context $\Delta$ and a substitution $\sigma$ such that: (i) $\Delta\vdash s\sigma\approx t\sigma$, for each $s\approx^?t\in\uprob$; (ii) $\Delta\vdash a\# t\sigma$, for each $s\#^?t\in \uprob$; and (iii) $\Delta\vdash X\sigma\approx X\sigma\sigma$, for all $X\in \var{\uprob}$. Derivability in items (i)-(iii) is established by the rules in Figure~\ref{fig:freshness_rules}.

We are interested in nominal equational unification problems, that is, in solving equations in the ground term algebra when in addition to $\alpha$-equivalence and freshness constraints, one is also interested in equational theories (such as commutativity, associativity, etc). Nominal equational unification problems, denoted $\eprob{E}$, are finite sets of constraints $s\approx_{\tt E}^?t$ and $a\#^? t$. Solutions of $\eprob{E}$ are defined in terms of $\approx_{\tt E}$, as expected, and rely on extending the rules in Figure~\ref{fig:freshness_rules} with the appropriate rules for the equational theory ${\tt E}$. For example, in the case of the commutative theory {\tt C}: first we need to specify the commutative function symbol $\tf{f^C}$ and the axiom that establishes the commutativity for this symbol  ${\tt C}=\{~\vdash \tf{f^C}(X,Y)=\tf{f^C}(Y,X)\}$. Now we extend and modify the rules of \Cref{fig:freshness_rules} as follows: (i) replace $\approx$ for $\approx_{\tt C}$ in all the rules;
(ii) add the rule \inferrule{\Delta\vdash t_1\approx_{\tt C}s_{i}\\ \Delta\vdash t_2\approx_{\tt C}s_{3-i}}{\Delta\vdash \tf{f^C}(t_1,t_2)\approx_{\tt C} \tf{f^C}(s_1,s_2)}. Together, these rules implement $\alpha,{\tt C}$-equality. 

Nominal unification algorithms are defined in terms of simplification rules that can be seen as bottom-up applications of derivation rules defining $\approx$ and $\#$. More details can be found in~\cite{DBLP:journals/mscs/Ayala-RinconSFS21}.



 \begin{example}\label{exa:comm}Consider  the nominal {\tt C}-unification problem $\eprob{C}=\{\tf{f^C}(X,Y)\approx_{\tt C}^? \tf{f^C}(c,(a \ b) \cdot  X)\}$. This problem reduces to two  simpler problems: 
\begin{align*}
& \mathbin{\rotatebox[origin=r]{30}{$\Longrightarrow$}} \{X\approx_{\tt C}^?c, Y \approx_{\tt C}^? (a\ b)\cdot X\} \\
\eprob{C} &  \Longrightarrow\{Y\approx_{\tt C}^?c, X \approx_{\tt C}^? (a\ b)\cdot X \}
\end{align*}
The top branch has a solution $\{X\mapsto c, Y\mapsto c\}$. The bottom branch is more problematic due to the fixed-point equation 
$\{X\approx_{\tt C}^? (a \ b) \cdot  X\}$.

We have seen in previous works on nominal unification (without equational theories) that the usual simplification rule for  fixed-point equations on variables is 
 $\pi\cdot X\approx \pi'\cdot X\Longrightarrow ds(\pi,\pi') \# X$. 
 However, this rule is not complete for the commutative case,   since we lose solutions such as $X\mapsto \tf{f^C}(a,b)$, $X\mapsto \tf{f^C}(\tf{f^C}(a,b), \tf{f^C}(a,b))$, $X\mapsto \forall [c]g(\tf{f^C}(a,b), c)$, etc. We have shown in~\cite{DBLP:conf/frocos/Ayala-RinconSFN17} that there are infinite independent solutions for this problem. 

\end{example}
Currently, there are two approaches to solving nominal equational problems:
\begin{enumerate}
    \item Do not attempt to solve fixed-point equations and leave them as part of solutions~\cite{DBLP:journals/mscs/Ayala-RinconSFS21}. Thus, solutions of unification problems would be triples $(\Delta, \sigma, P)$, where $P$ is a finite set of fixed-point equations of the form $X\approx_{E}\pi\cdot X$, and we can use extension of \Cref{fig:freshness_rules} to validate solutions. 
    \item Use a different and more expressive formalism, using fixed-point constraints that capture fixed-point equations~\cite{DBLP:journals/lmcs/Ayala-RinconFN19}. Solutions to unification problems are again pairs $(\Upsilon,\sigma)$, but with fixed-point contexts, and solutions are validated using the rules in~\Cref{fig:equality-rules1} taking into account the axioms defining the theory ${\tt E}$.
\end{enumerate}

First, with the approach (i)  the solutions to Example~\ref{exa:comm} are expressed by the triples $(\emptyset, \{X\mapsto c, Y\mapsto c\}, \emptyset)$ and $(\emptyset, \{Y\mapsto c\}, \{X\approx_{\tt C}^? (a \ b) \cdot  X\})$. Second, with the approach (ii) the solutions to the problem in Example~\ref{exa:comm} are expressed by the pairs $(\emptyset, \{X\mapsto c, Y\mapsto c\})$ and $(\{(a \ b) \fix{}  X\}, \{Y\mapsto c\}\})$, and concrete instances $\delta$ of $X$ that validate $(a \ b) \fix{}  X$, when composed with $\{Y\mapsto c\}$ (i.e. $\{Y\mapsto c\}\delta$), are also solutions to the initial problem.



\subsection{A sound fixed-point approach for commutativity: group generator}

By Theorem~\ref{thm:soundness_fix}, fixed-point constraints are a sound approach for proving validity of judgements in strong models, and these do not accept permutative theories, such as commutativity. But we can reason modulo the theories defined by the axioms in Example~\ref{ex:strong_axioms}.

The use of more powerful rules proposed in \Cref{ssec:group_gen} might be useful for commutativity, although limited in terms of derivable judgements. First we consider the following modification of the rules of \Cref{fig:equality-rules1}:
\begin{itemize}
    \item Consider the {\bf (ax)} rule with the axiom ${\sf C}=\{\tf{f^C}(X,Y)=\tf{f^C}(Y,X)\}$, where $\tf{f^C}\in \Sigma$;
    \item Replace the {\bf ($\fix{}$var)} rule for {\sc Gvar}$_{\fix{}}$.
\end{itemize}

The soundness is trivial, as it is a direct consequence of Theorem~\ref{thm:soundness_fix-group-generated}. 



\subsection{Another  sound fixed-point approach for commutativity: strong judgements}
The use of strong judgements as proposed in \Cref{ssec:strong_judgements} also seems promising for dealing with commutativity. First we consider the following modification of the rules of \Cref{fig:fixed-rules_new2}:
\begin{itemize}
    \item Replace $\falphaeq{}$ for $\falphaeq{C}$ in all the rules.
    \item Replace $\fix{}$ for $\fix{\tt C}$ in all the rules.
    \item Add the following rule for commutativity
    \inferrule{\new \lin{c}.\stupsilon{A}{\lin{c_0}}\vdash t_1\falphaeq{C}s_i\\ \new \lin{c}.\stupsilon{A}{\lin{c_0}}\vdash t_2\falphaeq{C}s_{3-i}}{\new \lin{c}.\stupsilon{A}{\lin{c_0}}\vdash \tf{f^C}(t_1,t_2)\falphaeq{C}\tf{f^C}(s_1,s_2)}
\end{itemize}





We conjecture that we can extend the translations  defined in \Cref{ssec:strong_judgements} to take into account the commutativity theory in order to obtain soundness by extending Theorem~\ref{thm:inherited soundness}. The expected extension of the translation would map judgements of the form  $\new\lin{c}.\stupsilon{A}{\lin{c_0}}\vdash (a\ b)\fix{\tt C} X$ where $b$ is not $\new$-quantified (and that are not equivalent to a freshness judgments) to a fixed-point equation $(a \ b)\act X \falphaeq{\alpha,\C} X$. (cf. Example~\ref{ex:fix_not_fresh}). This provides a foundation for unification algorithms modulo $\alpha$ and $\C$.
\section{Conclusion}
\label{sec:conclusion}

We formalised the intentional semantics for fixed-point constraints and verified that strong nominal sets are a sound denotational semantics for nominal algebra with fixed-point constraints. We proposed alternatives to regain soundness for the whole class of nominal sets, by proposing more powerful rules or stronger judgements and new rules. Our developments provide a sound approach for solving nominal equational problems modulo strong theories. However, the best approach to deal with non-strong theories, such as the permutative theories, still have to be investigated. Investigations towards completeness are under development and will be reported in future work.


\bibliographystyle{entics}
\bibliography{references}

\end{document}